\documentclass[10pt,onecolumn,twoside]{article}
  
 \usepackage{citesort}%
 \usepackage{rotating,multirow}
 \usepackage{comment}
 \usepackage{fullpage}
 \usepackage{slashbox}
 \usepackage{subfigure,cite,eurosym} 
 \usepackage{amsthm,multirow}
 \usepackage{graphicx}
 \theoremstyle{plain}
 \newtheorem{Prop}{Proposition}

 \usepackage{algorithmic}
\usepackage{url} 
 \usepackage{color}
\usepackage{soul}

\input alphabet

\def\barr{\begin{array}}                \def\earr{\end{array}}
\def\bcc{\begin{center}}                \def\ecc{\end{center}}
\def\cl#1{\centerline{#1}}
\def\bdoc{\begin{document}}             \def\edoc{\end{document}}
\def\ben{\begin{enumerate}}             \def\een{\end{enumerate}}
\def\beqn{\begin{eqnarray}}             \def\eeqn{\end{eqnarray}}
\def\beqnl#1{\beqn\label{#1}}           \def\eeqnl#1{\label{#1}\eeqn}
\def\beqnx{\begin{eqnarray*}}           \def\eeqnx{\end{eqnarray*}}
\def\bseqn{\begin{subeqnarray}}         \def\eseqn{\end{subeqnarray}}
\def\beq#1\eeq{\begin{equation}#1\end{equation}}
\def\bal#1\eal{\begin{align}#1\end{align}}
\def\balx#1\ealx{\begin{align*}#1\end{align*}}
\def\beqx{$$}                           \def\eeqx{$$}
\def\bfig{\protect\begin{figure}}       \def\efig{\protect\end{figure}}
\def\bfigx{\protect\begin{figure*}}     \def\efigx{\protect\end{figure*}}
\def\bfigt{\protect\begin{figurette}}   \def\efigt{\protect\end{figurette}}
\def\bfl{\begin{flushleft}}             \def\efl{\end{flushleft}}
\def\bfr{\begin{flushright}}            \def\efr{\end{flushright}}
\def\bit{\begin{itemize}}               \def\eit{\end{itemize}}
\def\bmi{\begin{minipage}}              \def\emi{\end{minipage}}
\newcommand{\bfmi}{\begin{fminipage}}   \newcommand{\efmi}{\end{fminipage}}
\def\bpic{\begin{picture}}              \def\epic{\end{picture}}
\def\btabl{\begin{table}}               \def\etabl{\end{table}}
\def\btab{\begin{tabular}} 
\def\btabu{\begin{tabular}}             \def\etabu{\end{tabular}}
\def\btabx{\begin{tabular*}}            \def\etabx{\end{tabular*}}

\input jidef
\def\Card{\#}
\def\dd{\,d}
\def\leq{\le}
\def\geq{\ge}
\renewcommand{\det}[1]{\bars{#1}}
\def\rang{{\mathrm{rank}}}
\def\V{.}              
\def\e#1{\rm{e}^{#1}}  
\def\numero{n\ensuremath{^{\circ}}}

\def\et{{\normalfont and }}
\def\editorname{(ed.)}
\def\editornames{(eds.)}
\def\eqname{Eq.~}
\def\eqnames{Eqs~}
\def\tabname{Table~}
\def\tabnames{Tables~}
\def\figname{Figure~}
\def\fignames{Figures~}
\def\chapname{Chapter~}
\def\chapnames{Chapters~}
\def\parname{\S~}
\def\parnames{\S~}
\def\sectname{Section~}
\def\sectnames{Sections~}
\def\theoname{Theorem~}
\newtheorem{theorem}{Theorem}[section]
\def\gui#1{``#1''}
\def\<{``}\def\>{''}
\def\ier{st\XS}
\def\iere{st\XS}
\def\ieme{th\XS}

\def\SI{\text{if\:}}       \def\Si{\text{If\:}}
\def\ET{\text{and\:}}       \def\OU{\text{or\:}}
\def\ALORS{\text{then\:}}
\def\DOU{\text{hence\:}}  \def\Ou{\text{where\:}}
\def\QUAND{\text{when\:}}
\def\POUR{\text{for\:}}   \def\POURTOUT{\text{for all\:}}
\def\SC{\text{u.\,c.\:}}   \def\SOUSC{\text{under constraints\:}}
\def\SINON{\text{otherwise}}
\def\AVEC{\text{with\:}}
\def\DANS{\text{in\:}}
\def\IFF{\textit{if and only if}\XS}
\def\ssi{\textit{si et seulement si}\XS}

\def\post{\textit{posterior}\XS}
\def\Post{\textit{Posterior}\XS}
\def\etal{{\itshape et al.}}
\def\etcoll{\textit{et al.}}
\def\andname{and\xspace}
\def\Toeplitz{Toeplitz\XS}
\def\wrt{w.r.t.\ }

\def\croext#1{\left(#1\right)}             \def\stdcroext#1{(#1)}
\def\bigcroext#1{\bigl(#1\bigr)}           \def\biggcroext#1{\biggl(#1\biggr)}
\def\Bigcroext#1{\Bigl(#1\Bigr)}           \def\Biggcroext#1{\Biggl(#1\Biggr)}

\newcommand{\DRAFT}[1]{#1}
\newcommand{\FINAL}[1]{}
\def\factor{0.58}

\definecolor{darkred}{rgb}{.5,0,0}
\newcommand{\JI}[1]{{\textcolor{blue}{#1}}}
\newcommand{\comJI}[1]{\emph{\textcolor{blue}{#1}}}
\newcommand{\delJI}[1]{\setstcolor{blue}\st{#1}}
\newcommand{\cor}[1]{\textcolor{darkred}{#1}}
\newcommand{\said}[1]{\textcolor{magenta}{\sf \small Sa\"id~:  #1}}
\newcommand{\com}[1]{{\emph{#1}}}
\newcommand{\rep}[1]{{\newline \textcolor{darkred}{{\bf Reply:} #1}}}
\newcommand{\xbold}{\ensuremath{\underline \xb}}
\newcommand{\xbnew}{\ensuremath{\overline \xb}}
\begin{document}
\title{Efficient Gaussian Sampling for Solving Large-Scale Inverse Problems using MCMC Methods\footnote{This work was supported by the french CNRS and the R\'egion des Pays de la Loire, France.}}

 \author{Cl\'ement~Gilavert, Sa\"id Moussaoui and J\'er\^ome~Idier \\[3mm]
Ecole Centrale Nantes, IRCCyN,  CNRS UMR 6597,\\
1 rue de la No\"e, 44321, Nantes Cedex 3, France.\\[3mm]
{\normalsize Corresponding author: \url{said.moussaoui@ec-nantes.fr}}}
\maketitle

\begin{abstract}
The resolution of many large-scale inverse problems using MCMC methods requires a step of drawing samples from a high dimensional Gaussian distribution. While direct Gaussian sampling techniques, such as those based on Cholesky factorization, induce an excessive numerical complexity and memory requirement, sequential coordinate sampling methods present a low rate of convergence. Based on the reversible jump Markov chain framework, this paper proposes an efficient Gaussian sampling algorithm having a reduced computation cost and memory usage. The main feature of the algorithm is to perform an approximate resolution of a linear system with a truncation level adjusted using a self-tuning adaptive scheme allowing to achieve the minimal computation cost. The connection between this algorithm and some existing strategies is discussed and its efficiency is illustrated on a linear inverse problem of image resolution enhancement.
\paragraph{}\emph{This paper is under revision before publication in IEEE Transactions on Signal Processing}
\end{abstract}
%
%
\section{Introduction} \label{sec:intro}
A common inverse problem arising in several signal and image processing applications is to recover a hidden object $\xb \in \eR^N$ (\eg an image or a signal) from a set of measurements $\yb \in \eR^M$ given an observation model~\cite{Bertero98,Idier08book}. The most frequent case is that of a linear model between \xb and \yb according to
\begin{equation} \label{Eq_LinMod}
\yb =  \Hb \xb + \nb,
\end{equation}
with $\Hb\in\eR^{M\times N}$ the known observation matrix and $\nb$ an additive noise term representing measurement errors and model uncertainties. Such a linear model covers many real problems such as, for instance, denoising~\cite{Frieden75}, deblurring~\cite{Demoment89}, and reconstruction from projections~\cite{Gordon71,Lewitt03}. 

The statistical estimation of $\xb$ in a Bayesian simulation framework~\cite{Gilks99,Robert01} firstly requires the formulation of the posterior distribution~$P(\xb,\Theta|\yb)$, with $\Theta$ a set of unknown hyper-parameters. Pseudo-random samples of \xb are then drawn from this posterior distribution. Finally, a Bayesian estimator is computed from these samples. Other quantities of interest, such as posterior variances, can be estimated likewise. Within the \emph{standard Monte Carlo} framework, independent realizations of the posterior law must be generated, which is rarely possible in realistic cases of inverse problems. One rather resorts to \emph{Markov Chain Monte Carlo} (MCMC) schemes, where Markovian dependencies between successive samples are allowed. A very usual sampling scheme is then to iteratively draw realizations from the conditional posterior densities $P(\Theta|\xb,\yb)$ and $P(\xb|\Theta,\yb)$, according to a \emph{Gibbs sampler} \cite{Geman84}.

In such a context, when independent Gaussian models $\Nc\left(\mub_y,\Rb_y\right)$ and $\Nc\left(\mub_x,\Rb_x\right)$ are assigned to the noise statistics and to the unknown object distribution, respectively, the set of hyper-parameters $\Theta$ determines the mean and the covariance of the latter two distributions. This statistical model also covers the case of priors based on hierarchical or latent Gaussian models such as Gaussian scale mixtures~\cite{andrews1974scale,champagnat2004connection} and Gaussian Markov random fields~\cite{Geman84,papandreou2010gaussian}. The additional parameters  of such models are then included in $\Theta$. According to this Bayesian modeling, the conditional posterior distribution $P(\xb|\Theta,\yb)$ is also Gaussian, $\Nc\left(\mub,\Qb\M\right)$, with a precision matrix $\Qb$ (\ie the inverse of the covariance matrix \Rb) given by
\beq
\label{eq_Q}
\Qb=\Hb\T\Rb_y^{-1}\Hb+\Rb_x^{-1},
\eeq
and a mean vector $\mub$ such that:
\beq
\label{eq_Qmu}
\Qb\mub=\Hb\T\Rb_y^{-1}(\yb-\mub_y)+\Rb_x^{-1} \mub_x.
\eeq
Let us remark that the precision matrix $\Qb$ generally depends on the hyper-parameter set $\Theta$ through $\Rb_y$ and $\Rb_x$, so that $\Qb$ is a varying matrix along the Gibbs sampler iterations. Moreover, the mean vector~$\mub$ is expressed as the solution of a linear system where \Qb is the normal matrix. 

In order to draw samples from the conditional posterior distribution $P(\xb|\Theta,\yb)$, a usual way is to firstly perform the Cholesky factorization of the covariance matrix~\cite{Scheurer62,Barr72}. Since equation \eqref{eq_Q} yields the precision matrix \Qb rather than the covariance matrix $\Rb$, Rue~\cite{rue2001fast} proposed to compute the Cholesky decomposition of \Qb, \ie $\Qb=\Cb_q\Cb_q\T$, and to solve the triangular system $\Cb_q\T\xb=\omegab$, where $\omegab$ is a vector of independent Gaussian variables of zero mean and unit variance. Moreover, the Cholesky factorization is exploited to calculate the mean \mub from \eqref{eq_Qmu} by solving two triangular systems sequentially. However, the Cholesky factorization of \Qb generally requires $\mathcal{O}(N^3)$ operations. Spending such a numerical cost at each iteration of the sampling scheme rapidly becomes prohibitive for large values of $N$. In specific cases where \Qb belongs to certain families of structured matrices, the factorization can be obtained with a reduced numerical complexity, \eg 
$\mathcal{O}(N^2)$ when \Qb is Toeplitz \cite{trench1964algorithm} or even $\mathcal{O}(N\log N)$ when $\Qb$ is circulant~\cite{geman1995nonlinear}. Sparse matrices can be also factored at a reduced cost~\cite{rue2001fast,Lalanne01}. Alternative approaches to the Cholesky factorization are based on using an iterative method for the calculation of the inverse square root matrix of  $\Qb$ using Krylov subspace methods~\cite{parker2012sampling,aune2013iterative,Chow14}. In practice, even in such favorable cases, the factorization often remains a burdensome operation to be performed at each iteration of the Gibbs sampler. 

The numerical bottleneck represented by the Cholesky factorization can be removed by using alternative schemes that bypass the step of exactly sampling $P(\xb|\Theta,\yb)$. For instance, a simple alternative solution is to sequentially sample each entry of \xb given the other variables according to a scalar Gibbs scheme \cite{Amit91}. However, such a scalar approach reveals extremely inefficient when $P(\xb|\Theta,\yb)$ is strongly correlated, since each conditional sampling step will produce a move of very small variance. As a consequence, a huge number of iterations will be required to reach convergence. A better trade-off between the numerical cost of each iteration and the overall convergence speed of the sampler must be found. 

\paragraph{} In this paper, we focus on a two-step approach named \emph{Independent Factor Perturbation} in~\cite{papandreou2010gaussian} and \emph{Perturbation-Optimization} in~\cite{orieux2012sampling} (see also~\cite{tan2010efficient,Lalanne01}).
It consists in
\bit
\item drawing a sample $\etab$ from $\Nc\left(\Qb\mub, \Qb\right)$,
\item solving the linear system $\Qb\xb = \etab$.
\eit
It can be easily checked that, when the linear system is solved exactly, the new sample $\xb$ is distributed according to $\Nc\left(\mub, \Qb\M\right)$. Hereafter, we refer to this method as \emph{Exact Perturbation Optimization} (E-PO). However, the numerical cost of E-PO is typically as high as the Cholesky factorization of \Qb. Therefore, an essential element of the Perturbation Optimization approach is to truncate the linear system solving by running  a limited number of iterations of an iterative algorithm such as the conjugate gradient method (CG)~\cite{papandreou2010gaussian,orieux2012sampling,tan2010efficient}. For the sake of clarity, let us call the resulting version \emph{Truncated Perturbation Optimization} (T-PO).

Skipping from E-PO to T-PO allows to strongly reduce the numerical cost of each iteration. However, let us stress that no convergence analysis of T-PO exists, to our best knowledge. It is only argued that a well-chosen truncation level  will induce a significant reduction of the numerical cost and a small error on \xb. The way the latter error alters the convergence towards the target distribution remains a fully open issue, that has not been discussed in existing contributions. Moreover, how the resolution accuracy should be chosen in practice is also an open question. 

A first contribution of the present paper is to bring practical evidence that the T-PO algorithm does not necessarily converge towards the target distribution (see Section~\ref{sec:TPOvsRJPO}). In practice, the implicit trade-off within T-PO is between the computational cost and the error induced on the target distribution, depending on the adopted truncation level. Our second contribution is to propose a new scheme similar to T-PO, but with a guarantee of convergence to the target distribution, whatever the truncation level. We call the resulting scheme \emph{Reversible Jump Perturbation Optimization} (RJPO), since it incorporates an accept-reject step derived within the Reversible Jump MCMC (RJ-MCMC) framework~\cite{Green95,waagepetersen2001tutorial}. Let us stress here that the numerical cost of the proposed test is marginal, so that RJPO has nearly the same cost per iteration as T-PO. Finally, we propose an unsupervised tuning of the truncation level allowing to automatically achieve a pre-specified overall acceptance rate or even to minimize the computation cost at a constant effective sample size. The resulting algorithm can be viewed as an adaptive (or controlled) MCMC sampler~\cite{Andrieu01,andrieu2008tutorial,atchade2009adaptive}. 
  
\paragraph{} The rest of the paper is organized as follows: Section~\ref{sec:revjump} introduces the global framework of RJ-MCMC and presents a general scheme to generate Gaussian vectors. Section~\ref{sec:algos} considers a specific application of the previous results, which finally boils down to the proposed RJPO sampler. Section~\ref{sec:TPOvsRJPO} analyses the performance of RJPO compared to T-PO on simple toy problems and presents the adaptive RJPO which incorporates an automatic control of the truncation level. Finally, in section~\ref{sec:appli}, an example of linear inverse problem, the unsupervised image resolution enhancement is presented to illustrate the applicability of the method. These results show the superiority of the RJPO algorithm over the usual Cholesky factorization based approaches in terms of computational cost and memory usage. 
%
%
\section{The reversible jump MCMC framework} \label{sec:revjump}
The sampling procedure consists on  constructing a Markov chain whose distribution asymptotically converges to the target distribution~$P_\Xb(\cdot)$. Let $\xbold \in \eR^{N}$ be the current sample of the Markov chain and $\xbnew$ the new sample obtained according to a transition kernel derived in the reversible jump framework.

\subsection{General framework}  
In the constant dimension case, the Reversible Jump MCMC strategy \cite{Green95,waagepetersen2001tutorial} introduces an auxiliary variable $\zb \in \eR^{L}$, obtained from a distribution $P_\Zb(\zb|\xbold)$ and a deterministic move according to a  differentiable transformation 
\begin{alignat*}{2}
\phib: & \left(\eR^{N} \times \eR^{L}\right)  \mapsto \left (\eR^{N} \times \eR^{L}\right) \\
&  (\xbold, \zb)  \mapsto  (\xb, \sb)
\end{alignat*}
This transformation must also be reversible, that is $\phib(\xb,\sb)=(\xbold, \zb)$. The new sample $\xbnew$ is thereby obtained by submitting $\xb$ (resulting from the deterministic move) to an accept-reject step with an acceptance probability given by
\[
\alpha (\xbold,\xb| \zb)=\min\left(1,\frac{P_\Xb(\xb)P_\Zb(\sb|\xb)}{P_\Xb(\xbold)P_\Zb(\zb|\xbold)}|J_{\phib}(\xbold,\zb)|\right),
\]
with $J_{\phib}(\xbold,\zb)$ the Jacobian determinant of the transformation $\phib$ at $(\xbold,\zb)$.

Actually, the choice of the conditional distribution $P_\Zb(\cdot)$ and the transformation $\phib(\cdot)$ must be adapted to the target distribution $P_\Xb(\cdot)$ and affects the resulting Markov chain properties in terms of correlation and convergence rate. 

\subsection{Gaussian case}To sample from a Gaussian distribution~$\xb\sim\Nc\left(\mub,\Qb\M\right)$, we generalize the scheme adopted in~\cite{deForcrand1999monte}. We set $L=N$ and take an auxiliary variable $\zb \in \eR^{N}$ distributed according to \begin{equation}\label{eq:genauxvar}
P_\Zb(\zb|\xbold)=\Nc\left(\Ab\xbold+\bb, \Bb\right),
\end{equation}
where \Ab, \Bb and \bb denote a $N\times N$ real matrix, a  $N\times N$ real positive definite matrix and a  $N\times 1$ real vector, respectively.
The choice of the latter three quantities will be discussed later. The proposed deterministic move is performed using the transformation $\phib$ such that 
\begin{equation} \label{eq:genmove}
\left(
\begin{array}{c}
   	\xb\\
	\sb
\end{array}
\right)
=
\left(
\begin{array}{c}
   	 \phib_1(\xbold,\zb)\\
	 \phib_2(\xbold,\zb)
\end{array}
\right)
=
\left(
\begin{array}{c}
   	 -\xbold+\fb(\zb)\\
	 \zb
\end{array}
\right),
\end{equation}
with functions $\left(\phib_1: (\eR^{N} \times \eR^{N}) \mapsto \eR^{N}\right)$ and $\left(\phib_2: (\eR^{N} \times \eR^{N}) \mapsto \eR^{N}\right)$ and $\left(\fb : \eR^{N} \mapsto \eR^{N}\right)$.

\begin{Prop}\label{th_1} 
Let an auxiliary variable $\zb$ be obtained according to~\eqref{eq:genauxvar} and a proposed sample $\xb$ resulting from \eqref{eq:genmove}. Then the acceptance probability is 
\beq
\label{eq_alpha}
\alpha(\xbold,\xb | \zb) = \min\left(1,e^{-\rb(\zb)\T\left(\xbold-\xb\right)}\right),
\eeq
with
\begin{equation}
\label{eq:fz2}
\rb(\zb)=\Qb\mub+\Ab\T\Bb^{-1}\left(\zb-\bb\right)-\frac{1}{2}\left(\Qb+\Ab\T\Bb^{-1}\Ab\right)\fb(\zb).
\end{equation}
In particular, the acceptance probability equals one when $\fb(\zb)$ is defined as the exact solution of the linear system
\begin{equation}
\label{eq:fz}
\dfrac{1}{2}\left(\Qb+\Ab\T\Bb^{-1}\Ab\right) \fb(\zb) = \Qb\mub+\Ab\T\Bb^{-1}\left(\zb-\bb\right).
\end{equation}
\end{Prop}
\begin{proof}
See appendix~\ref{annex_deltaS}.
\end{proof}
Let us emphasize that \bb is a dummy parameter, since the residual $\rb(\zb)$ (and thus $\alpha(\xbold,\xb|\zb)$) depends on \bb through $\zb-\bb$ only. However, choosing a specific expression of \bb jointly with \Ab and \Bb will lead to a simplified expression of $\rb(\zb)$ in the next section.

Proposition~\ref{th_1} plays a central role in our proposal. When the exact resolution of \eqref{eq:fz} is numerically costly, it allows to derive a procedure where the resolution is performed only approximately, at the expense of a lowered acceptance probability. The conjugate gradient algorithm stopped before convergence, is a typical example of an efficient tool allowing to approximately solve \eqref{eq:fz}.
\begin{Prop}\label{th_3}
Let an auxiliary variable $\zb$ be obtained according to~\eqref{eq:genauxvar}, a proposed sample $\xb$ resulting from \eqref{eq:genmove} and $\fb(\zb)$ be the exact solution of \eqref{eq:fz}. The correlation between two successive samples is zero if and only if matrices \Ab and \Bb are chosen such that
\beq
\label{ABA}
\Ab\T\Bb\M\Ab=\Qb.
\eeq
\end{Prop}
\begin{proof}
See Appendix~\ref{annex_corr}.
\end{proof}

Many couples $(\Ab,\Bb)$ fulfill condition~\eqref{ABA}
\begin{itemize}
\item Consider the Cholesky factorization $\Qb=\Cb_q\Cb_q\T$ and take $\Ab=\Cb_q\T$, $\Bb=\Ib$. It leads to $\zb=\Cb_q\T\xbold + \bb+\omegab$ with $\omegab\sim\Nc\left(\zerob,\Ib_{N}\right)$. According to \eqref{eq:fz}, the next sample $\xbnew = - \xbold + \fb(\zb)$, will be obtained as 
\begin{align*}
\xbnew &=-\xbold+\left(\Cb_q\Cb_q\T\right)\M\left(\Qb\mub+\Cb_q(\zb-\bb)\right), \\
&=\left(\Cb_q\T\right)\M \left(\Cb_q\M\Qb\mub+\omegab\right).
\end{align*}
Such an update scheme is exactly the same as the one proposed by Rue in~\cite{rue2001fast}. 
\item The particular configuration
\beq
\label{choice}
\Ab=\Bb=\Qb\quad \text{and}\quad \bb=\Qb \mub.
\eeq
is retained in the sequel, since: 
\bit
\item[$i)$]~$\Ab\T\Bb\M\Ab=\Qb$ is a condition of Proposition~\ref{th_3}, 
\item[$ii)$]~$\bb=\Qb\mub$ simplifies equation~\eqref{eq:fz} to a linear system $\Qb \fb(\zb)=\zb$. 
\eit
In particular, it allows to make a clear connection between our RJ-MCMC approach and the E-PO algorithm in the case of an exact resolution of the linear system.
\end{itemize}
%
%
\section{Gaussian Sampling in the Reversible Jump MCMC framework} \label{sec:algos}

The resulting algorithm for the sampling of a Gaussian distribution in the RJMCMC framework is related to the sampling of an auxiliary variable according to~\eqref{eq:genauxvar} and the resolution of the linear system~\eqref{eq:fz}.  

\subsection{Sampling the Auxiliary Variable}\label{sec:eta}

According to \eqref{choice}, the auxiliary variable $\zb$ is distributed according to $\Nc\left(\Qb\xbold+\Qb\mub,\Qb\right)$. It can then be expressed as $\zb=\Qb\xbold+\etab$, $\etab$ being distributed according to $\Nc\left(\Qb\mub, \Qb\right)$. Consequently, the auxiliary variable sampling step is reduced to the simulation of $\etab$, which is the perturbation step in the PO algorithm. 

In~\cite{papandreou2010gaussian,orieux2012sampling}, a subtle way of sampling $\etab$ is proposed. It consists in exploiting equation \eqref{eq_Qmu} and \emph{perturbing} each factor separately:
\begin{enumerate}
\item Sample $\etab_y\sim\Nc\left(\yb-\mub_y, \Rb_y\right)$,
\item Sample $\etab_x\sim\Nc\left(\mub_x, \Rb_x\right)$,
\item Set $\etab=\Hb\T\Rb_y\M\etab_y+\Rb_x\M\etab_x$, a sample of $\Nc\left(\Qb\mub,\Qb\right)$.
\end{enumerate}
It is important to notice that such a tricky method is interesting since matrices $\Rb_y$ and $\Rb_x$ have often a simple structure if not diagonal. 
 
We emphasize that this perturbation step can be applied more generally for the sampling of any Gaussian distribution, for which a factored expression of the precision matrix $\Qb$ is available under the form~$\Qb=\Fb\T \Fb$, with matrix $\Fb\in \eR^{N'\times N}$. In such a case, $\etab=\Qb\mub + \Fb\T \wb$, where $\omegab\sim\Nc\left(\zerob,\Ib_{N'}\right)$.
 
\subsection{Exact resolution case}
As stated by proposition~\ref{th_1}, the exact resolution of system~\eqref{eq:fz} implies an acceptance probability of one. The resulting sampling procedure is thus based on the following steps:
\begin{enumerate}
\item Sample $\etab\sim\Nc\left(\Qb\mub, \Qb\right)$,
\item Set $\zb=\Qb\xbold+\etab$,
\item Take $\xbnew=-\xbold+\Qb^{-1}\zb$.
\end{enumerate}

Let us remark that $\xbnew=-\xbold+\Qb\M(\Qb\xbold+\etab)=\Qb\M\etab$, so the handling of variable \zb can be skipped and Steps 2 and 3 can be merged to an equivalent but more direct step: 
\begin{enumerate}
\stepcounter{enumi}
\item Set $\xbnew=\Qb^{-1}\etab$.
\end{enumerate}
In the exact resolution case, the obtained algorithm is thus identical to the E-PO algorithm~\cite{orieux2012sampling}. According to Proposition~\ref{th_3}, E-PO enjoys the property that each sample is totally independent from the previous ones. However, a drawback is that the exact resolution of the linear system $\Qb\xb=\etab$ often leads to an excessive numerical complexity and memory usage in high dimensions~\cite{papandreou2010gaussian}. In practice, early stopping of an iterative solver such as the linear conjugate gradient algorithm is used, yielding the Truncated Perturbation Optimization (T-PO) version. The main point is that, up to our knowledge, there is no theoretical analysis of the efficiency of T-PO and of its convergence to the target distribution. Indeed, the simulation tests provided in Section~\ref{sec:TPOvsRJPO} indicate that convergence to the target distribution is not guaranteed. As shown in the next subsection, two slight but decisive modifications of T-PO lead us to the RJPO version, which is a provably convergent algorithm.

\subsection{Approximate resolution case}
In the case \eqref{choice}, equation~\eqref{eq:fz2} reduces to
\begin{equation} \label{eq:fz3}
\rb(\zb)=\zb-\Qb\fb(\zb).
\end{equation}
Therefore, a first version of the RJPO algorithm is as follows:
\begin{enumerate}
\item Sample $\etab\sim\Nc\left(\Qb\mub, \Qb\right)$,
\item \label{steplinsolv} Set $\zb=\Qb\xbold+\etab$. Solve the linear system $\Qb\ub=\zb$, in an approximate way. Let $\widehat \ub$ denote the obtained solution, $\rb(\zb)=\zb-\Qb\ub$ and propose $\widehat \xb = -\xbold + \widehat \ub$,
\item \label{stepaccept} With probability $\min\left(1,e^{-\rb(\zb)\T(\xbold-\widehat \xb)}\right)$, set $\xbnew=\widehat \xb$, otherwise set $\xbnew=\xbold$.
\end{enumerate}
An important point concerns the initialization of the linear solver in Step~\ref{steplinsolv}: in the case of an early stopping, the computed approximate solution may depend on the initial point $\ub_0$. On the other hand, $\fb(\zb)$ must not depend on $\xbold$, otherwise the reversibility of the deterministic move \eqref{eq:genmove} would not be ensured. Hence, the initial point $\ub_0$ must not depend on $\xbold$ either. In the rest of the paper, $\ub_0=\zerob$ is the default choice.

\paragraph{} A more compact and direct version of the sampler can be obtained by substituting $\xb=\fb(\zb)-\xbold$ in equation \eqref{eq:fz3}. The latter reduces to the solving of the system $\Qb\xb=\etab$. Step~\ref{steplinsolv} of the RJPO algorithm is then simplified to:
\begin{enumerate}
\stepcounter{enumi}
\item Solve the linear system $\Qb\xb=\etab$ in an approximate way. Let $\widehat{\xb}$ denote the obtained solution and $\rb(\zb)=\etab-\Qb\widehat \xb$.
\end{enumerate}
For the reason just discussed above, the initial point $\xb_0$ of the linear solver must be such that $\ub_0=\xb_0+\xbold$ does not depend on $\xbold$. Hence, as counterintuitive as it may be, choices such as $\xb_0=\zerob$ or $\xb_0=\xbold$ are not allowed, while $\xb_0=-\xbold$ is the default choice corresponding to $\ub_0=\zerob$.

It is remarkable that both T-PO and the proposed algorithm (RJPO) rely on the approximate resolution of the same linear system $\Qb\xb=\etab$. However, RJPO algorithm incorporates two additional ingredients that make the difference in terms of mathematical validity:
\bit
\item RJPO relies on an accept-reject strategy to ensure the sampler convergence in the case of an approximate system solving,
\item There is a constraint on the initial point $\xb_0$ of the linear system solving: $\xb_0+\xbold$ must not depend on $\xbold$.
\eit

\subsection{Implementation issues}
There is no constraint on the choice of the linear solver, nor on the early stopping rule, except that they must not depend on the value of $\xbold$. Indeed, any linear system solver, or any quadratic programming method could be employed. In the sequel, we have adopted the linear conjugate gradient algorithm for two reasons:
\bit
\item Early stopping (\ie \emph{truncating}) the conjugate gradient iterations is a very usual procedure to approximately solve a linear system, with well-known convergence properties towards the exact solution \cite{Barrett94}. Moreover, a preconditioned conjugate gradient could well be used to accelerate the convergence speed.
\item It lends itself to a matrix-free implementation with reduced memory requirements, as far as matrix-vector products involving matrix \Qb can be performed without explictly manipulating such a matrix.
\eit
On the other hand, we have selected a usual stopping rule based on a threshold on the relative residual norm:
\beq
\label{RRN}
\epsilon = \dfrac{\|\etab - \Qb \xb\|_{2}}{\|\etab\|_{2}}.
\eeq
%
%
\section{Performance analysis}
\label{sec:TPOvsRJPO}
The aim of this section is to analyze the performance of the RJPO algorithm and to discuss the influence of the relative residual norm (and hence, the truncation level of the conjugate gradient (CG) algorithm) on the performances of the proposed RJPO algorithm. A second part is dedicated to the analysis of the optimal choice of the relative residual norm. 

\subsection{Incidence of the approximate resolution} 
Let us consider a Gaussian distribution with a precision matrix $\Qb=\Rb\M$ and a mean vector $\mub$ defined by 
 \begin{gather*}
R_{ij} = \sigma^{2}\rho^{|i-j|}, \quad (\forall i, j=1,\ldots, N) \\
\mu_i \sim \Uc[0, 10],   \quad (\forall i,\ldots, N) 
\end{gather*}
with $N=20$, $\sigma^{2}=1$ and $\rho=0.8$. 

\paragraph{} Figure \ref{fig:tpoillust} shows the distribution of $K=5000$ samples obtained by running the TPO algorithm for different truncation levels of a conjugate gradient algorithm. It can be noted that en early stopping, with $J<5$, leads to a sample distribution different from the target one. One can also notice that the related acceptance probability expressed in the RJMCMC framework suggests at least $J=6$ iterations to get samples with a nonzero acceptance probability. One can also that an exact resolution is not needed since the acceptance probability is almost equal to one after $J>10$ iterations. This result allows to conclude that the idea of truncating the system resolution is relevant, since it allows to avoid unnecessary calculations, but an acceptation-reject step must be added to ensure a correct behavior of the sampler.
\begin{figure}[h!]
\centering
\subfigure[$J=4$ iterations]{\includegraphics[width=4cm]{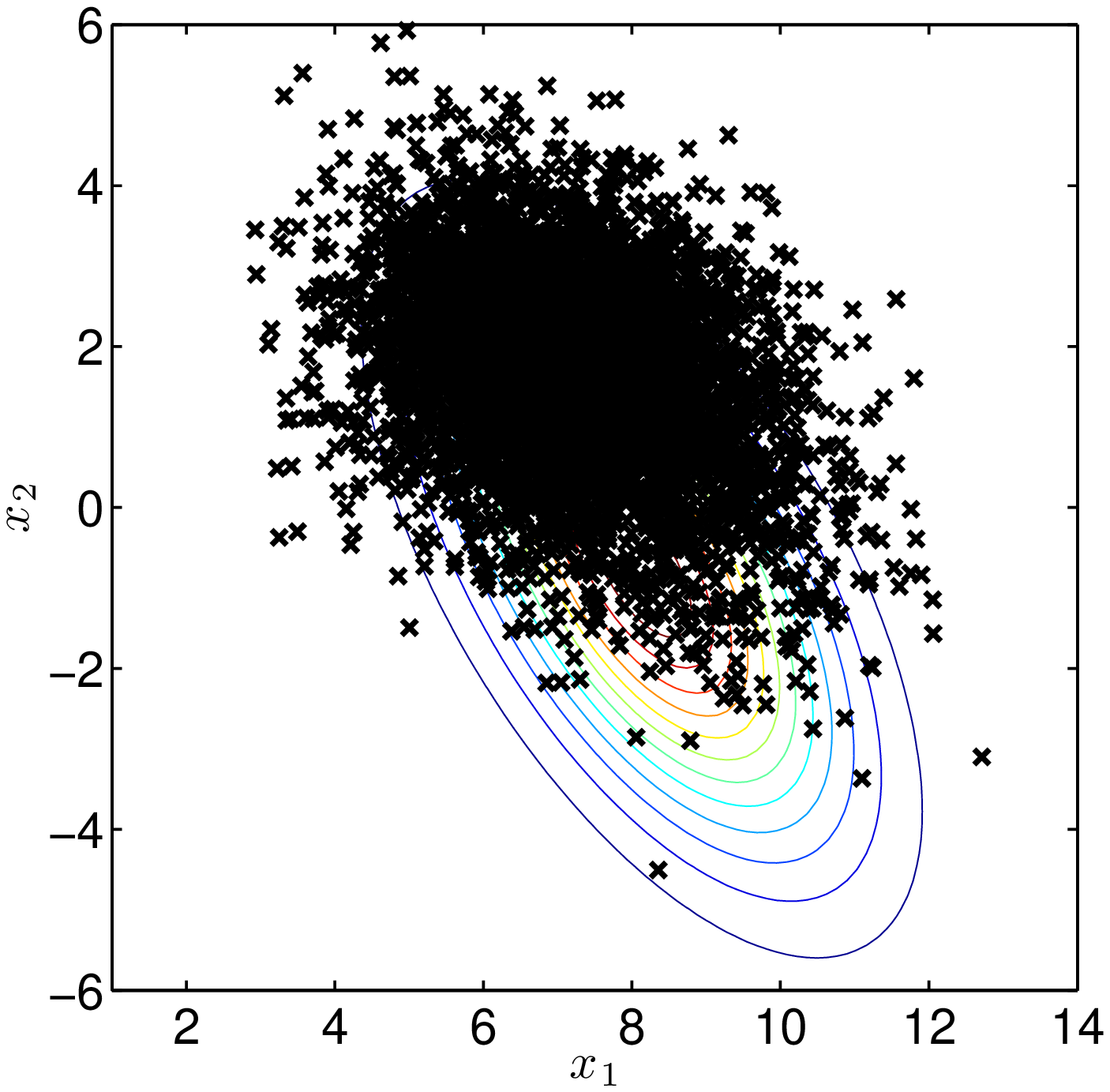}} 
\subfigure[$J=10$ iterations]{\includegraphics[width=4cm]{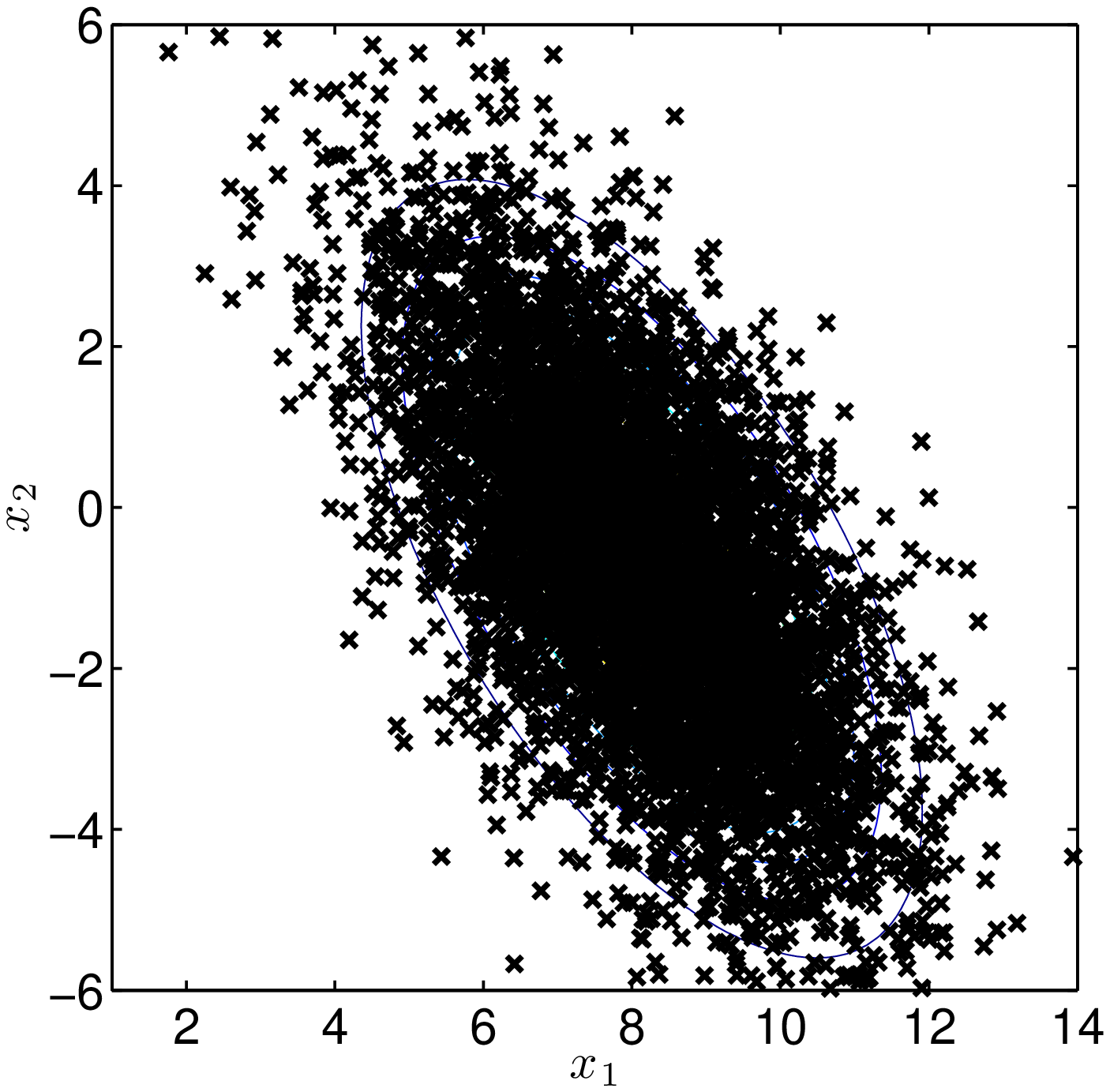}}
\subfigure[Sample mean]{\includegraphics[width=4cm]{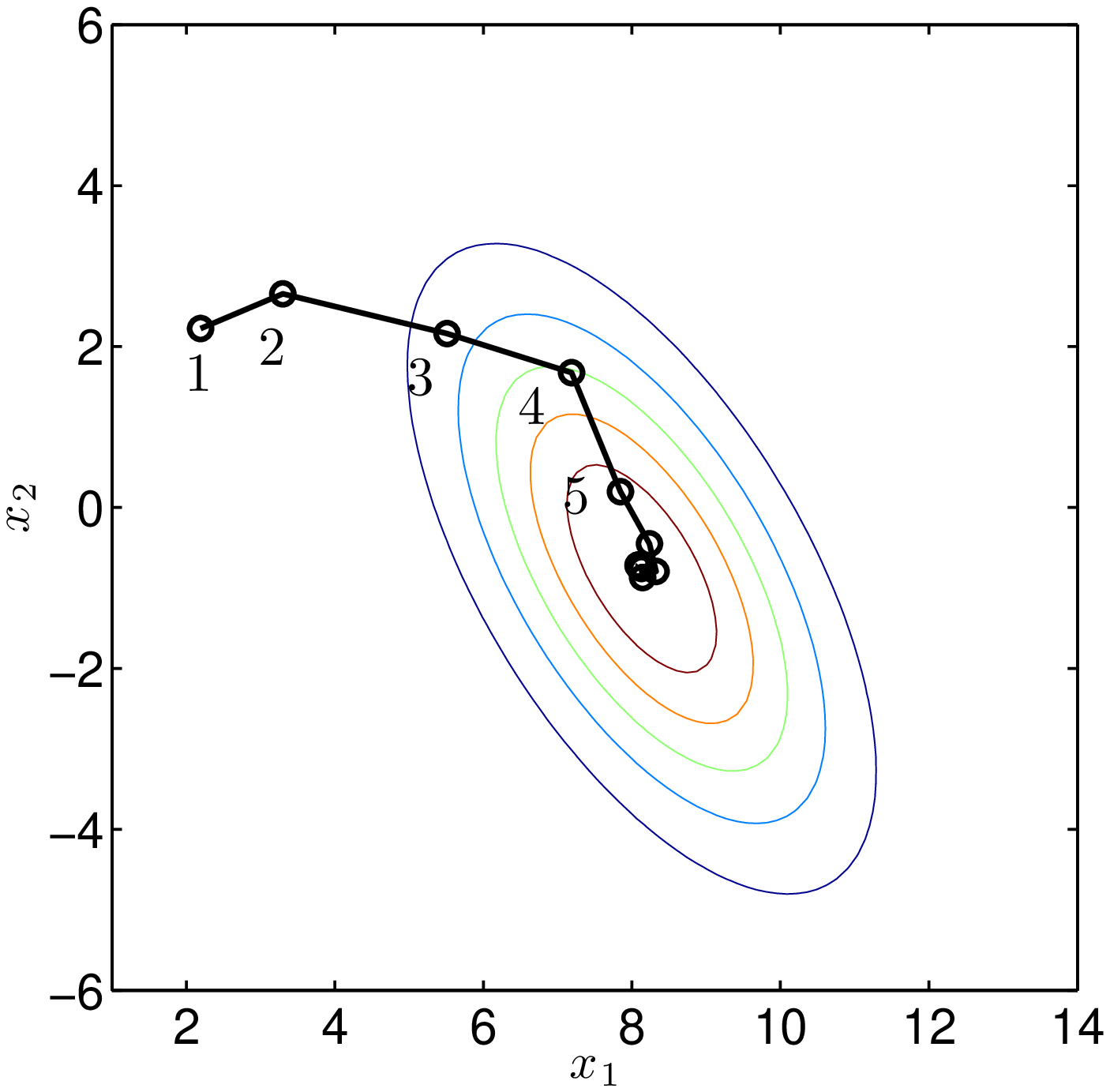}} 
\subfigure[Acceptance rate\label{fig:tpo-d}]{\includegraphics[width=4.1cm]{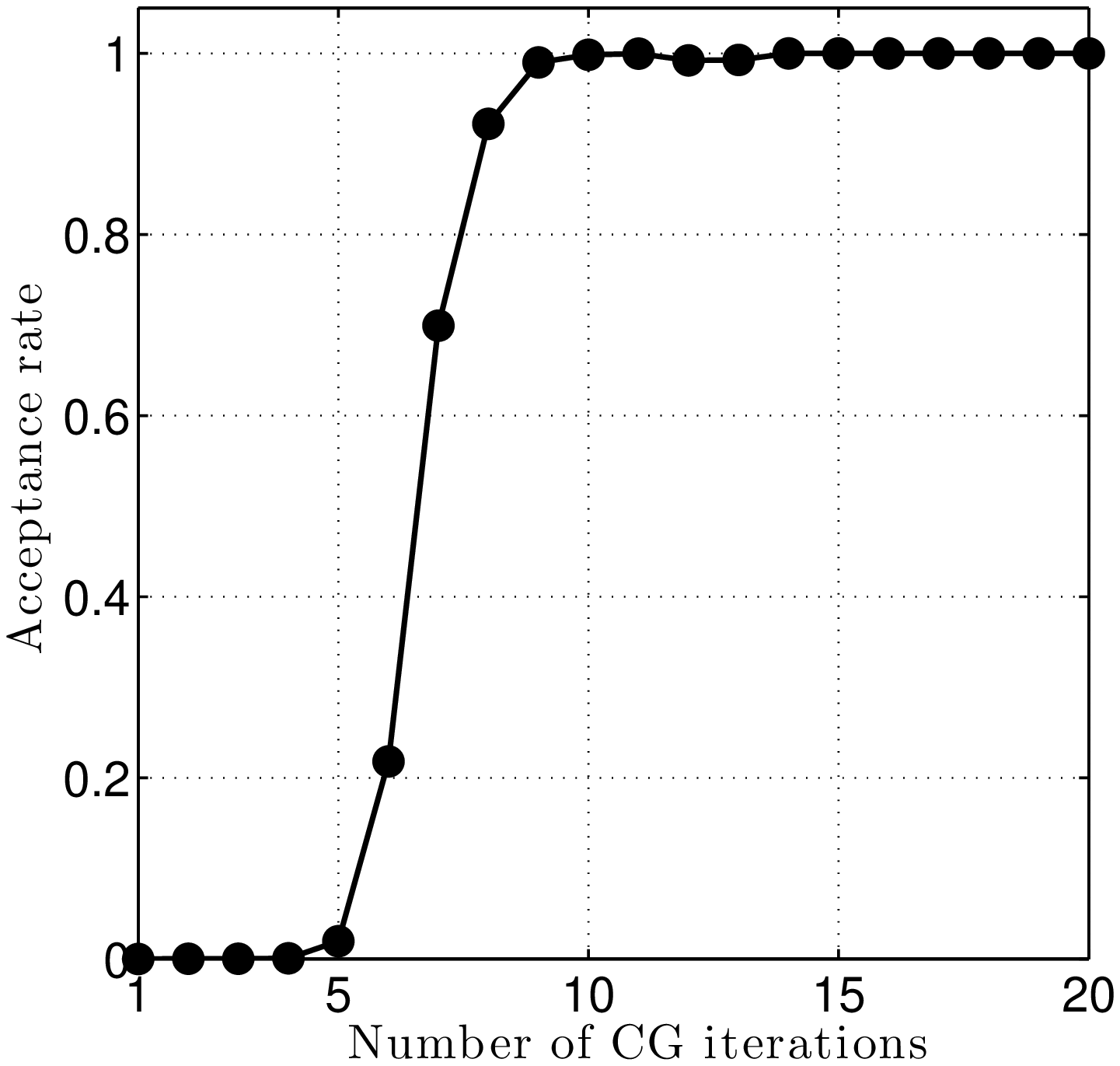}}
\caption{Influence of the truncation level on the distribution of $K=10^{5}$ samples obtained by the TPO algorithm.}
\label{fig:tpoillust}
\end{figure}

\subsection{Acceptance rate} 
We focus in this experiment on a small size problem ($N=16$) to discuss the influence of the truncation level on the numerical performance in terms of acceptance rate and estimation error. For the retained Gaussian sampling schemes, both RJPO and T-PO are run for a number of CG iterations allowing to reach a predefined value of relative residual norm~\eqref{RRN}. We also discuss the influence of the problem dimension on the best value of the truncation level leading to a minimal total number of CG iterations before convergence. 

\figref{fig:TvsRJ:alpha} illustrates the average acceptance probability obtained over $n_{\max}=10^{5}$  iterations of the RJPO sampler for different relative residual norm values. It can be noted that the acceptance rate is almost zero when the relative residual \cor{norm} is larger than $10^{-2}$ and monotonically increases for higher resolution accuracies. Moreover, a relative residual norm lower than $10^{-5}$ leads to an acceptance probability almost equal to one. Such a curve indicates that the stopping criterion of the CG must be chosen carefully in order to run the RJPO algorithm efficiently and to get non-zero acceptance probabilities. Finally, note that this curve mainly depends on the condition number of the precision matrix $\Qb$. Even if the shape of the acceptance curve stays the same for different problems, it happens to be difficult to determine the value of the relative residual norm that corresponds to a given acceptance rate.
\begin{figure}[h!]
     \centering
     \includegraphics[width=8cm,height=4.5cm]{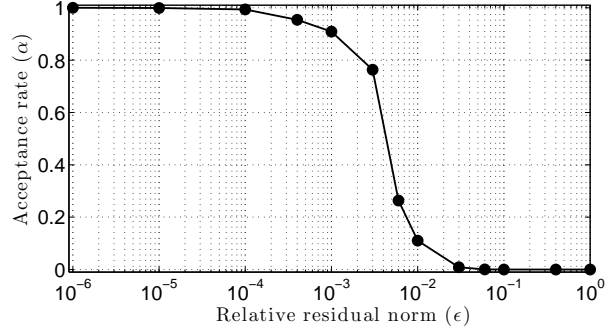}
    \caption{Acceptance rate of the RJPO algorithm for different values of the relative residual norm in a small size problem ($N=16$).}
    \label{fig:TvsRJ:alpha}
\end{figure}

\subsection{Estimation error} 
The estimation error is assessed as the relative mean square error (RMSE) on the estimated mean vector and covariance matrix using the Markov chain samples
\begin{equation}
\textrm{RMSE}({\mub}) = \dfrac{\|\mub - \hat \mub\|_{2}}{\|\mub\|_{2}} \quad \text{ and } \quad 
 \textrm{RMSE}({\Rb}) = \dfrac{\|\Rb - \hat \Rb\|_{F}}{\|\Rb\|_{F}}
\end{equation}
where $\|\cdot\|_{F}$ and $\|\cdot\|_{2}$ represent the Frobenius and the $\ell_2$ norms, respectively. $\mub$, $\Rb$, $\hat \mub$, and $\hat \Rb$ are respectively the mean and the covariance matrix of the Gaussian vector, and their empirical estimates using the generated Markov chain samples according to 
\begin{equation*}
\begin{cases}
\hat \mub = \dfrac{1}{n_{\max}-n_{\min}+1} \sum\limits_{n=n_{\min}}^{n_{\max}} \xb_{n} \\
\hat \Rb = \dfrac{1}{n_{\max}-n_{\min}} \sum\limits_{n=n_{\min}}^{n_{\max}} (\xb_{n}-\hat \mub)(\xb_{n}-\hat \mub)\T
\end{cases}
\end{equation*}
with $n_{\min}$ iterations of burn-in and $n_{\max}$ total iterations. 

\paragraph{} As expected, \figref{fig:TvsRJ:error} indicates that the estimation error is very high if the acceptance rate is zero (when the relative residual norm is lower than $10^{-2}$), even for RJPO after $n_{\max}=10^5$ iterations. This is due to the very low acceptance rate which slows down the chain convergence. However, as soon as new samples are accepted, RJPO leads to the same performance as when the system is solved exactly (E-PO algorithm). On the other hand, T-PO keeps a significant error for small and moderate resolution accuracies. Naturally, both methods present similar performance when the relative residual norm is very low since these methods tend to provide almost the same samples with an acceptance probability equal to one. This experimental result clearly highlights the deficiency of T-PO: the system must be solved with a relatively high accuracy to avoid an important estimation error. On the other hand, in the RJPO algorithm the acceptance rate is a good indicator whether the value of the relative residual norm threshold is appropriate to ensure a sufficient mixing of the chain.
\begin{figure}[h!]
    \centering
    \subfigure[Mean vector estimation]{\includegraphics[width=8cm,height=4.5cm]{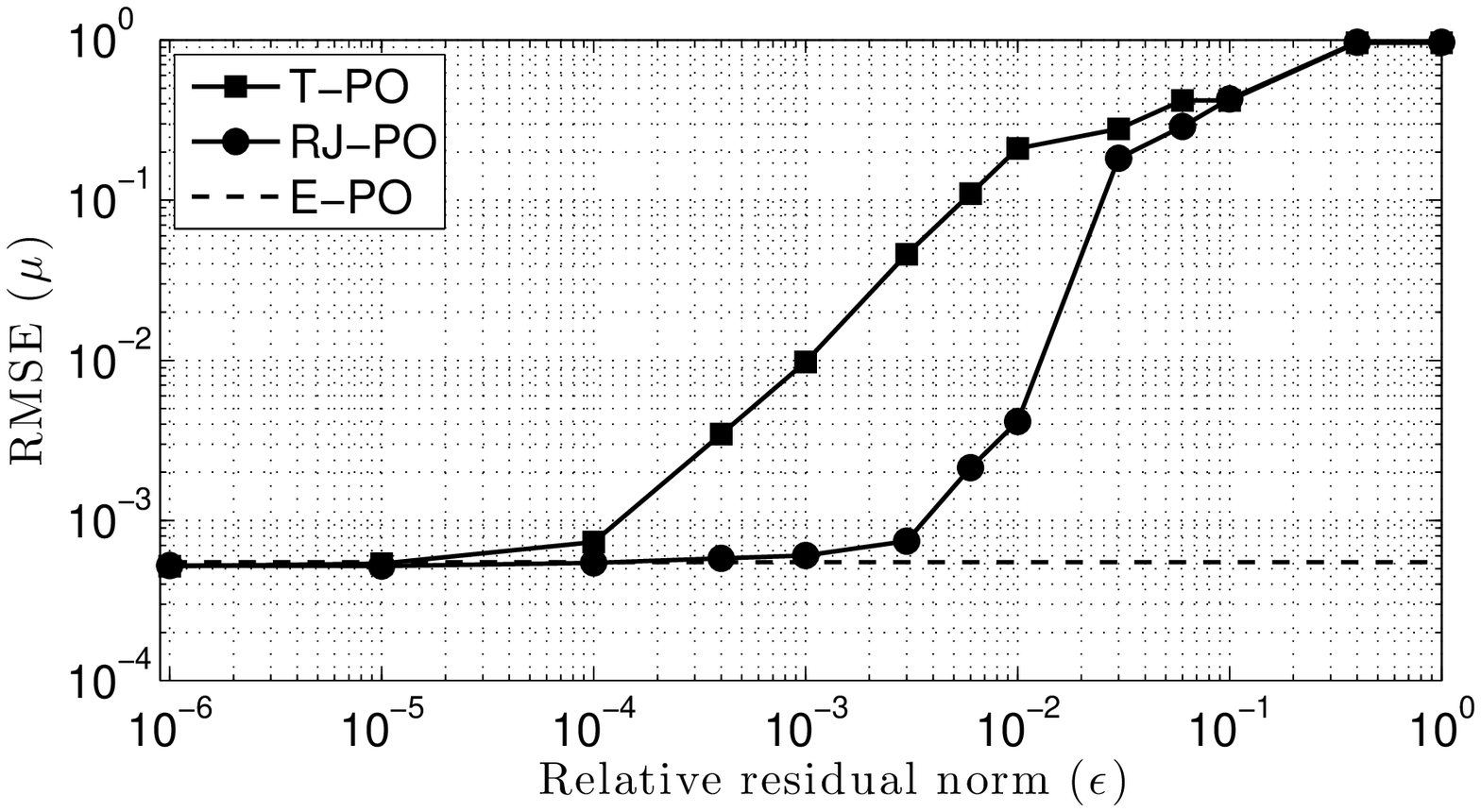}}
    \subfigure[Mean vector estimation]{\includegraphics[width=8cm,height=4.5cm]{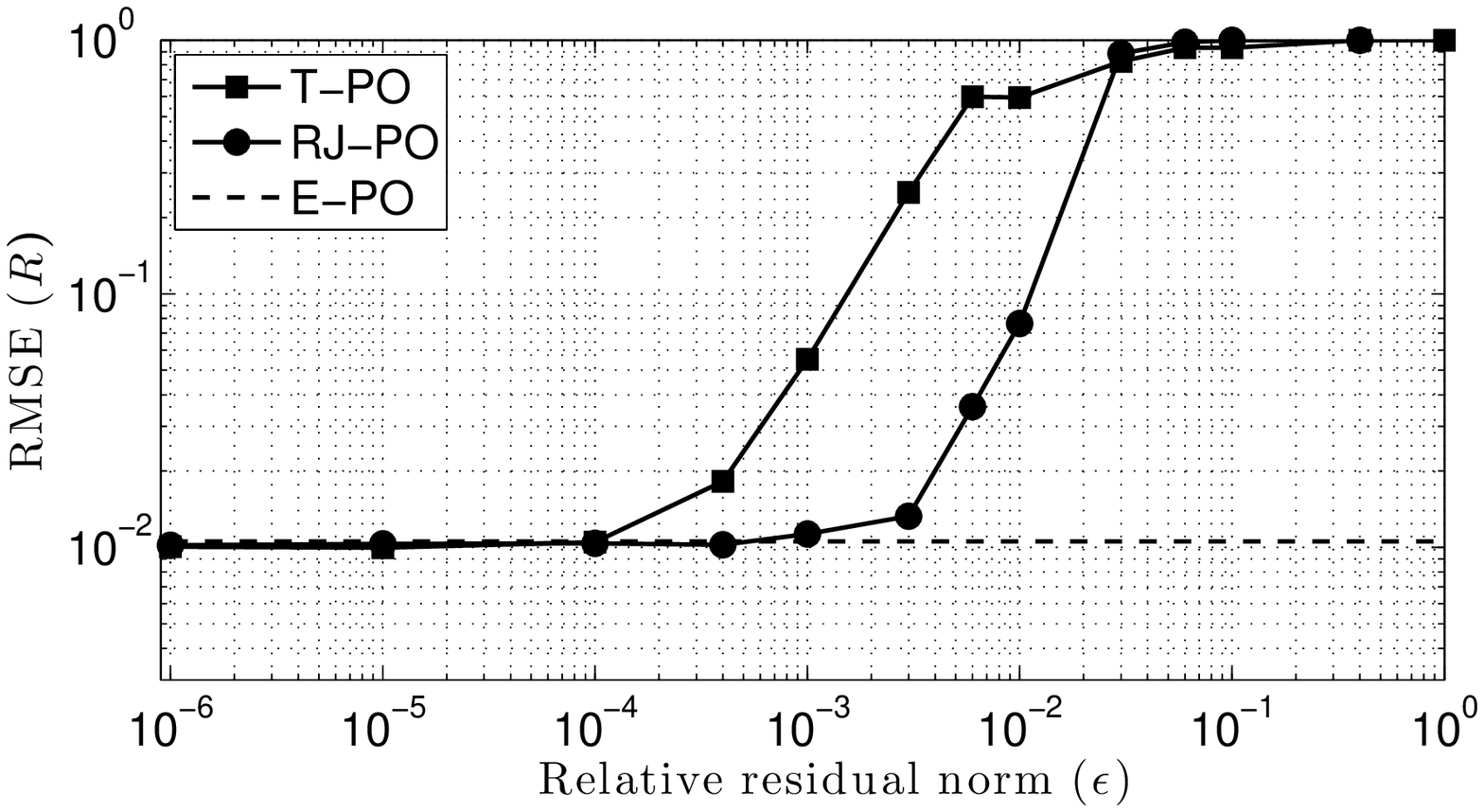}} 
    \caption{Estimation error for different values of the truncation level after $n_{\max}=10^5$ iterations of E-PO, T-PO and RJPO algorithms: (a) mean vector, (b) covariance matrix.
    \label{fig:TvsRJ:error}}
\end{figure}

\subsection{Computation cost} Since the CG iterations correspond to the only burdensome task, the numerical complexity of the sampler can be expressed in terms of the total number $J_{\textrm{tot}}$ of CG iterations to be performed before convergence and the number of required samples to get efficient empirical approximation of the estimators. 

To assess the Markov chain convergence, we first use the Gelman-Rubin criterion based on multiple chains~\cite{gelman1992inference}, which consists in computing a scale reduction factor based on the between and within-chain variances. In this experiment 100 parallel chains are considered. The results are summarized in \figref{fig:TvsRJ:CGiter}. It can be noted that a lower acceptance rate induces a higher number of iterations since the Markov chain converges more slowly towards its stationary distribution. One can also see that a minimal cost can be reached and, according to \figref{fig:TvsRJ:alpha}, it corresponds to an acceptance rate of almost one. As the acceptance rate decreases, even a little, the computational cost rises very quickly. Conversely, if the relative residual is too small, the computation effort per sample will decrease but additional sampling iterations will be needed before convergence, which naturally increases the overall computation cost. The latter result points out the need to appropriately choose the truncation level to jointly avoid a low acceptance probability and a high resolution accuracy of the linear system since both induce unnecessary additional computations. 
\begin{figure}[h!]
\centering
          \includegraphics[width=8.5cm]{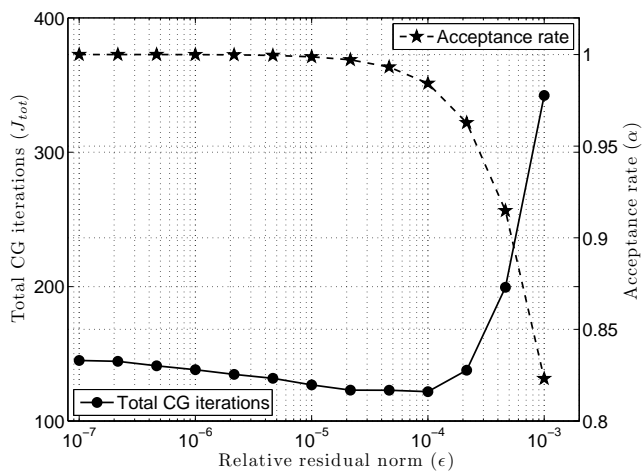}
          \caption{Number of CG iterations before convergence and acceptance probability of the RJPO algorithm for different values of relative residual norm for a small size problem ($N=16$).}
    \label{fig:TvsRJ:CGiter}
\end{figure}

\subsection{Statistical efficiency} 
The performance of the RJPO sampler can also be analyzed using the effective sample size (ESS)~\cite[p. 125]{Liu01}. This indicator gives the number of independent samples, $n_{\textrm{eff}}$, that would yield the same estimation variance of approximating the Bayesian estimator as $n_{\max}$ successive samples of the simulated chain~\cite{Goodman89}. It is related to the chain autocorrelation function according to
\beq
\label{neff}
n_{\textrm{eff}} = \dfrac{n_{\max}}{1 + 2 \sum\limits_{k=1}^{\infty} \rho_{k}}
\eeq
where $\rho_{k}$ the autocorrelation coefficient at lag $k$. In the Gaussian sampling context, such a relation allows to define how many iterations $n_{\max}$ are needed for each resolution accuracy to get chains having the same effective sample size. Under the hypothesis of a first-order autoregressive chain, $\rho_{k} = \rho^{k}$, so \eqref{neff} leads to the ESS ratio
\beq
\label{ESSR}
\textrm{ESSR} = \dfrac{n_{\textrm{eff}}}{n_{\max}} = \dfrac{1 - \rho}{1 +  \rho}.
\eeq
It can be noted that the ESSR is equal to one when the samples are independent $(\rho = 0)$ and decreases as the correlation between successive samples grows. In the RJPO case, we propose to define the \emph{computing cost per effective sample} (CCES) as
\beq
\label{CCES}
\textrm{CCES} = \dfrac{J_{\textrm{tot}}}{n_{\textrm{eff}}} =\dfrac{J}{\textrm{ESSR}}
\eeq
where $J=J_{\textrm{tot}}/n_{\max}$ is the average number of CG iterations per sample. \figref{fig:rjpoESS} shows the ESSR and the CCES in the case of a Gaussian vector of dimension $N=16$. It can be seen that an early stopped CG algorithm induces a very small ESSR, due to a large sample correlation value, and thus a high effective cost to produce accurate estimates. On the contrary, a very precise resolution of the linear system induces a larger number of CG iterations per sample but a shorter Markov chain since the ESSR is almost equal to 1. The best trade-off is produced by intermediate values of the relative residual
norm around $\epsilon=2\cdot10^{-4}$.
\begin{figure}[h!]
    \centering
     \includegraphics[width=8.5cm]{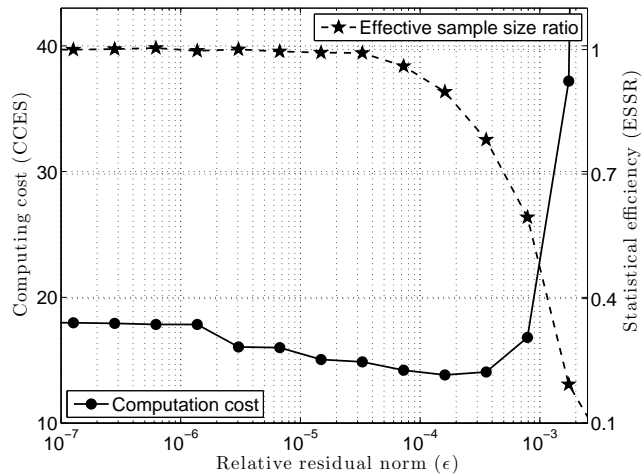} \\
    \caption{Computing cost per effective sample of the RJPO algorithm for different relative residual norm values on a small size problem $(N=16)$ estimated from $n_{\max}=10^4$ samples. }
    \label{fig:rjpoESS}
\end{figure}

To conclude, the Gelman-Rubin convergence diagnostic and the ESS approach both confirm that the computation cost of the RJPO can be reduced by appropriately truncating the CG iterations. Although the Gelman-Rubin convergence test is probably more accurate, since it is based on several independent chains, the CCES based test is far simpler and provides nearly the same trade-off in the tested example. Such results motivate the development of an adaptive strategy to automatically adjust the threshold parameter $\epsilon$ by tracking the minimizer of the CCES. The proposed strategy is presented in Subsection \ref{sec:adrjpo}.
   
\subsubsection{Influence of the dimension} \figref{fig:optimrjpo} summarizes the optimal values of the truncation level $\epsilon$ that allows to minimize the CCES for different values of $N$. The best trade-off is reached for decreasing values of $\epsilon$ as $N$ grows. More generally, the same observation can be made as the problem conditioning deteriorates. In practice, predicting the appropriate truncation level for a given problem is difficult. Fortunately, \figref{fig:optimrjpo} also indicates that the optimal setting is obtained for an acceptance probability that remains almost constant. The best trade-off is clearly obtained for an acceptance rate $\alpha$ lower than one ($\alpha=1$ corresponds to $\epsilon=0$, \ie to the exact solving of $\Qb\xb=\etab$). In the tested example, the optimal truncation level $\epsilon$ rather corresponds to an acceptance rate around $0.99$. However, finding an explicit mathematical correspondence between $\epsilon$ and $\alpha$ is not a simple task. In the next subsection, we propose an unsupervised tuning strategy of the relative residual norm allowing either to achieve a predefined target acceptance rate, or even to directly optimize the computing cost per effective sample.
\begin{figure}[h!]
    \centering
      \includegraphics[width=8.5cm]{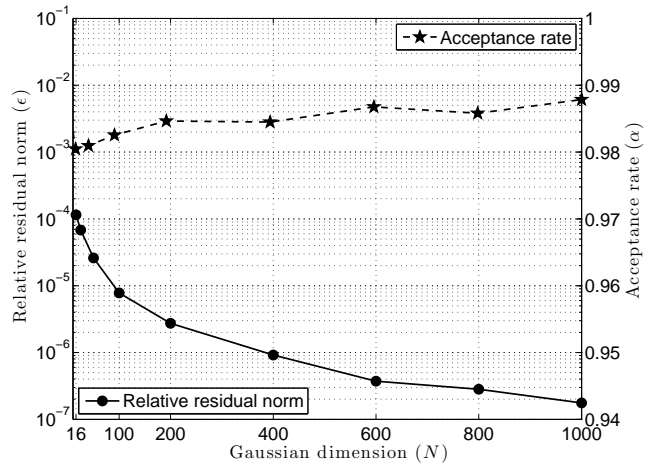} \\
    \caption{Influence of the problem dimension on the optimal values of the relative residual norm and the acceptance rate.}
    \label{fig:optimrjpo}
\end{figure}

\subsection{Adaptive tuning of the resolution accuracy} \label{sec:adrjpo}
The suited value of the relative residual norm $\epsilon$ to achieve a desired acceptance rate~$\alpha_t$ can be adjusted recursively using a Robbins-Monro type algorithm~\cite{bercu2012robbins}. The result sampling approach relies on the family of controlled/adaptive MCMC methods~\cite{Andrieu01}. See for instance \cite{andrieu2008tutorial} for a tutorial. Such an adaptive scheme is formulated in the stochastic approximation framework~\cite{benveniste2012adaptive} in order to solve a non-linear equation of the form $g(\theta)=0$ using an update 
 \begin{equation}
\theta_{n+1} = \theta_{n} + K_{n}\left[g(\theta_{n})+\nu_{n}\right]
\end{equation}
where $\nu$ is a random variable traducing the uncertainty on each evaluation of function $g(\cdot)$ and $\{K_{n}\}$ is a sequence of step-sizes ensuring stability and convergence~\cite{Andrieu03}. Such a procedure has been already used for the optimal scaling of adaptive MCMC algorithms~\cite{andrieu2008tutorial}. The use of Robbins-Monro procedure for the optimal scaling of some adaptive MCMC algorithms such as the Random walk Metropolis-Hastings (RWMH) algorithm. It is mainly shown that such procedure breakdown the Markovian structure of the chain but it does not alter its convergence towards the target distribution. For instance, it is used in~\cite{Haario01,Atchade05} to set adaptively the scale parameters of a RWMH algorithm in order to reach the optimal acceptance rate suggested by theoretical or empirical analysis~\cite{Roberts97,Gelman96}. The same procedure was also used by \cite{Atchade06} for the adaptive tuning of a Metropolis-adjusted Langevin algorithm (MALA) to reach the optimal acceptance rate proposed by \cite{Roberts98}.

\subsubsection{Achieving a target acceptance rate}
In order to ensure the positivity of the relative residual norm $\epsilon$, the update is performed on its logarithm. At each iteration $n$ of the sampler, the relative residual norm is \cor{adjusted} according to
\begin{equation}
\log \epsilon_{n+1} = \log \epsilon_{n} + K_{n} \left[\alpha(\xbold_{n},\xb_{n}) - \alpha_t\right].
\end{equation}
where $\alpha_t$ is a given target acceptance probability and $\{K_n\}$ is a sequence of step-sizes decaying to $0$ as $n$ grows in order to ensure the convergence of the Markov chain to the target distribution. As suggested in~\cite{andrieu2008tutorial}, the step-sizes are chosen according to $K_{n}= K_{0}/n^{\beta}$, with $\kappa \in ]0, 1]$. We emphasize that more sophisticated methods, such as those proposed in~\cite{bercu2012robbins} could be used to approximate the acceptance rate curve and to derive a more efficient adaptive strategy for choosing this parameter.

The adaptive RJPO is applied to the sampling of the previously described Gaussian distribution using the adopted step-size with parameters $K_{0} = 1$ and $\kappa = 0.5$. \figref{fig:RJPO:automatic} presents the evolution of the average acceptance probability and the obtained relative residual norm for three different values of the target acceptance rate $\alpha_t$. One can note that the average acceptance rate converges to the desired value. Moreover, the relative residual norm also converges to the expected values according to \figref{fig:TvsRJ:alpha} (for example, the necessary relative residual norm to get an acceptance probability \cor{$\alpha_{t}=0.8$} is equal to $1.5 \cdot 10^{-3}$).
\begin{figure}[h!]
\centering
     \begin{tabular}{cc}
             \includegraphics[width=8cm]{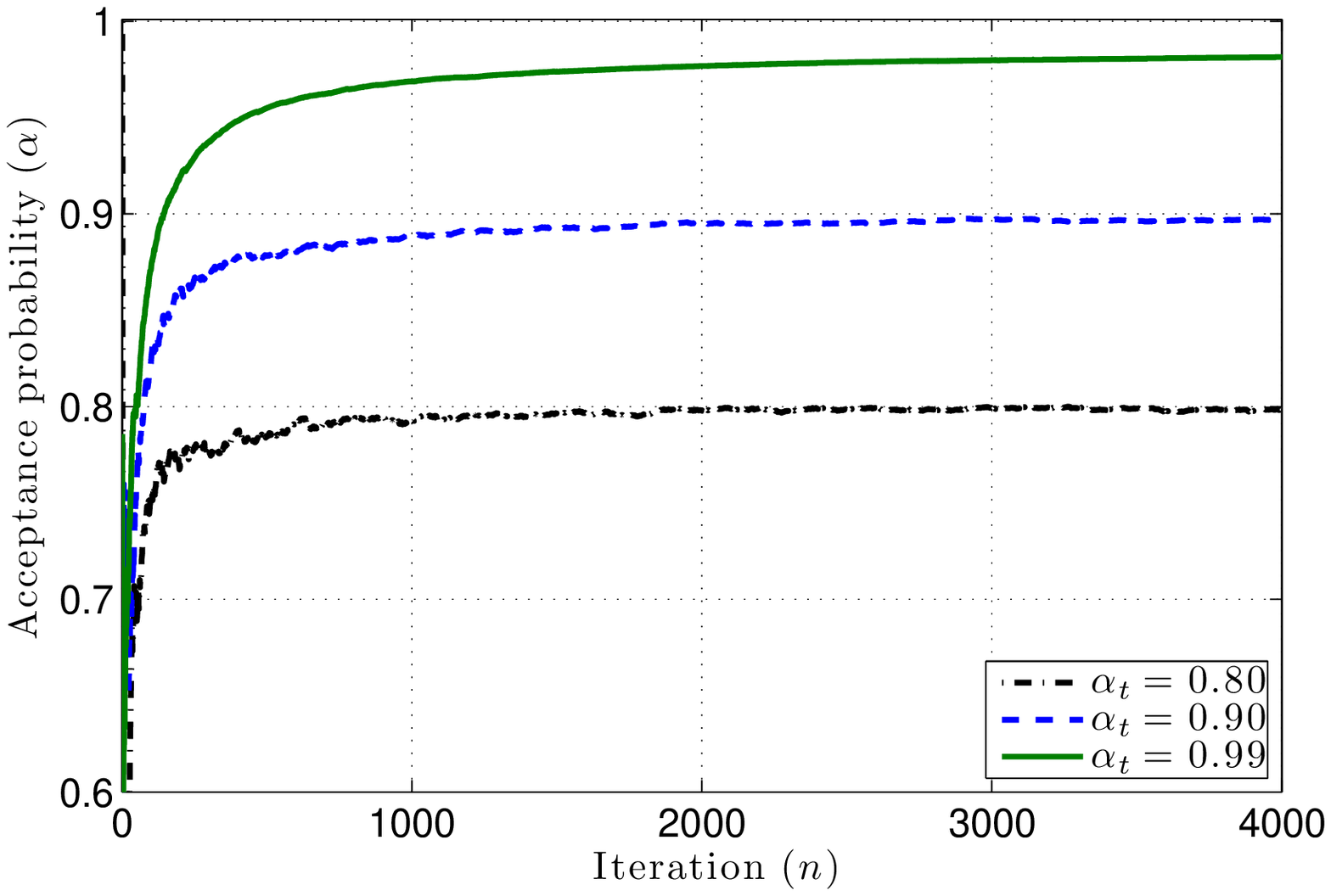} &
             \includegraphics[width=8cm]{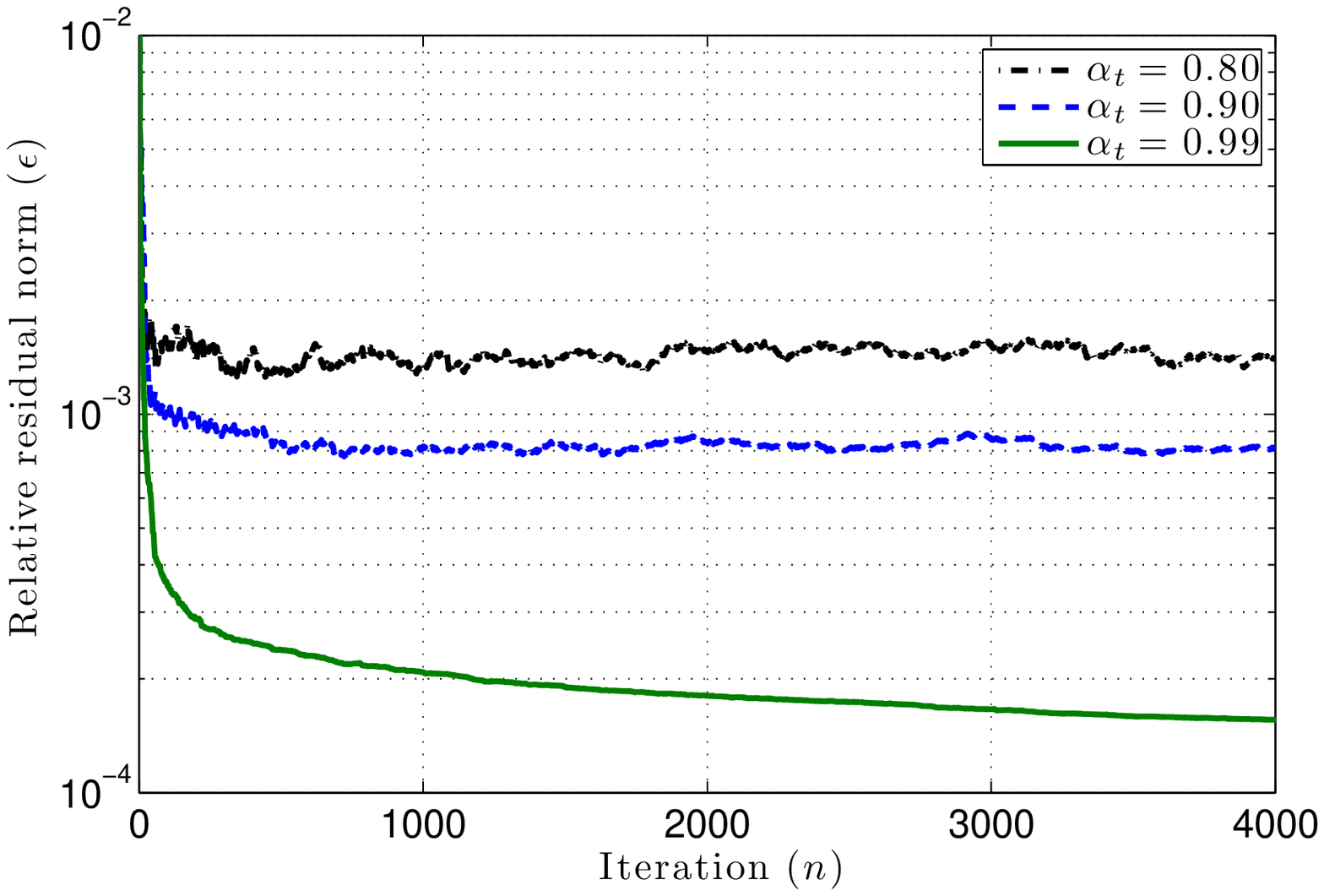} \\
                          {\small (a) Acceptance probability} & 
             {\small (b) Relative residual norm} \\
             \end{tabular}
\caption{Behavior of the adaptive RJPO for 1000 iterations and three values of the target acceptance probability: (a) Evolution of the average acceptance probability  and (b) Evolution of the computed relative residual norm. \label{fig:RJPO:automatic}}
\end{figure}
In practice, it remains difficult to \aprio determine which acceptance rate should be targeted to achieve the faster convergence. The next subsection proposes to modify the target of the adaptive strategy to directly minimize the CCES~\eqref{CCES}.

\subsubsection{Optimizing the numerical efficiency}
A given threshold $\epsilon$ on the relative residual norm induces an average truncation level $J$ and an ESSR value, from which the CCES can be deduced according to~\eqref{CCES}. Our goal is to adaptively adjust the threshold value $\epsilon$ in order to minimize the CECS. Let $J_{\textrm{opt}}$ be the average number of CG iterations per sample corresponding to the optimal threshold value. In the plane $(J,\textrm{ESSR)}$, it is easy to see that $J_{\textrm{opt}}$ is the abscissa of the point at which the tangent of the ESSR curve intercepts the origin (see \figref{fig:esscost}).

\paragraph{}The ESSR is expressed by \eqref{ESSR} as a function of the chain correlation $\rho$, the latter being an implicit function of the acceptance rate $\alpha$. For $\alpha=1$, $\rho=0$ according to Proposition~\ref{th_3}. For $\alpha=0$, $\rho=1$ since no new sample can be accepted. For intermediate values of $\alpha$, the correlation lies between 0 and 1, and it is typically decreasing. It can be decomposed on two terms:
\bit
\item With a probability $1-\alpha$, the accept-reject procedure produces identical (\ie maximally correlated) samples in case of rejection.
\item In case of acceptance, the new sample is slightly correlated with the previous one, because of the early stopping of the CG algorithm.
\eit
While it is easy to express the correlation induced by rejection, it is difficult to find an explicit expression for the correlation between accepted samples. However, we have checked that the latter source of correlation is negligible compared to the correlation induced by rejection. If we approximately assume that accepted samples are independent, we get $\rho=1-\alpha$, which implies
$$
\textrm{ESSR} = \dfrac{2-\alpha}{\alpha}  \Longrightarrow \textrm{CCES} = \dfrac{\alpha}{2-\alpha} \; J.
$$
Thus, by necessary condition, the best tuning of the relative residual norm leading to the lowest CCES is obtained by setting,
$$
J \, \dfrac{d \alpha}{dJ} - \alpha + \dfrac{\alpha^{2}}{2} = 0.
$$

\paragraph{} The stochastic approximation procedure is applied to adaptively adjust the optimal value of the relative residual norm according to
\begin{equation}
\log \epsilon_{n+1} = \log \epsilon_{n} + K_{n} \left(J_{n} \, \dfrac{d \alpha_{n}}{dJ} - \alpha_{n} + \dfrac{\alpha_{n}^{2}}{2} \right),
\end{equation}
where $ \dfrac{d \alpha_{n}}{dJ}$ is evaluated numerically. \figref{fig:adrjpo} illustrates that the proposed adaptive scheme efficiently adjusts $\epsilon$ to minimize the CCES, since $J_{\textrm{opt}}$ is around 25 according to \figref{fig:esscost}.
 \begin{figure}[h!]
    \centering
\subfigure[Efficiency\label{fig:esscost}]{\includegraphics[height=6cm]{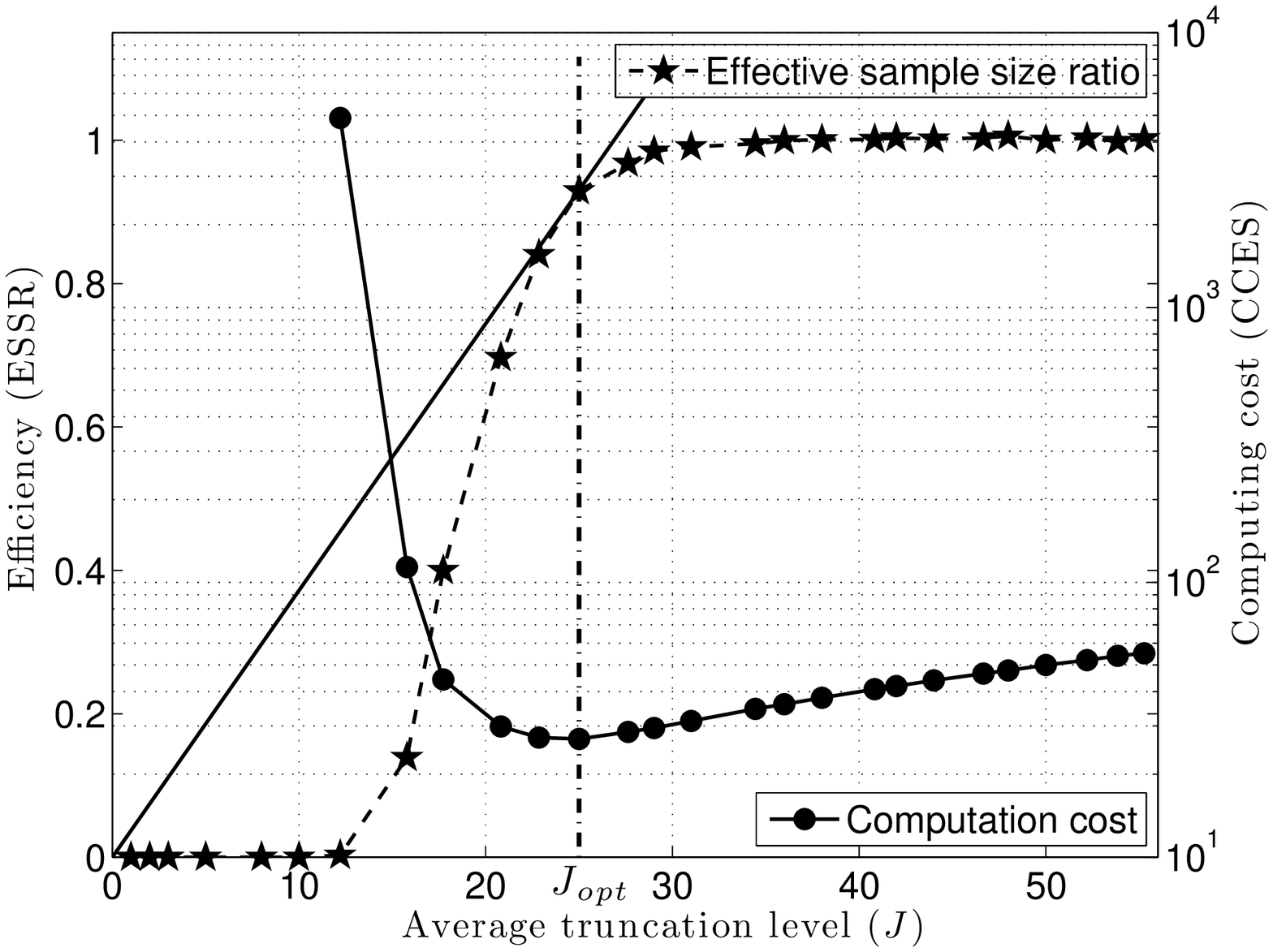}}
\subfigure[Adaptive tuning \label{fig:adrjpo}]{\includegraphics[height=6cm]{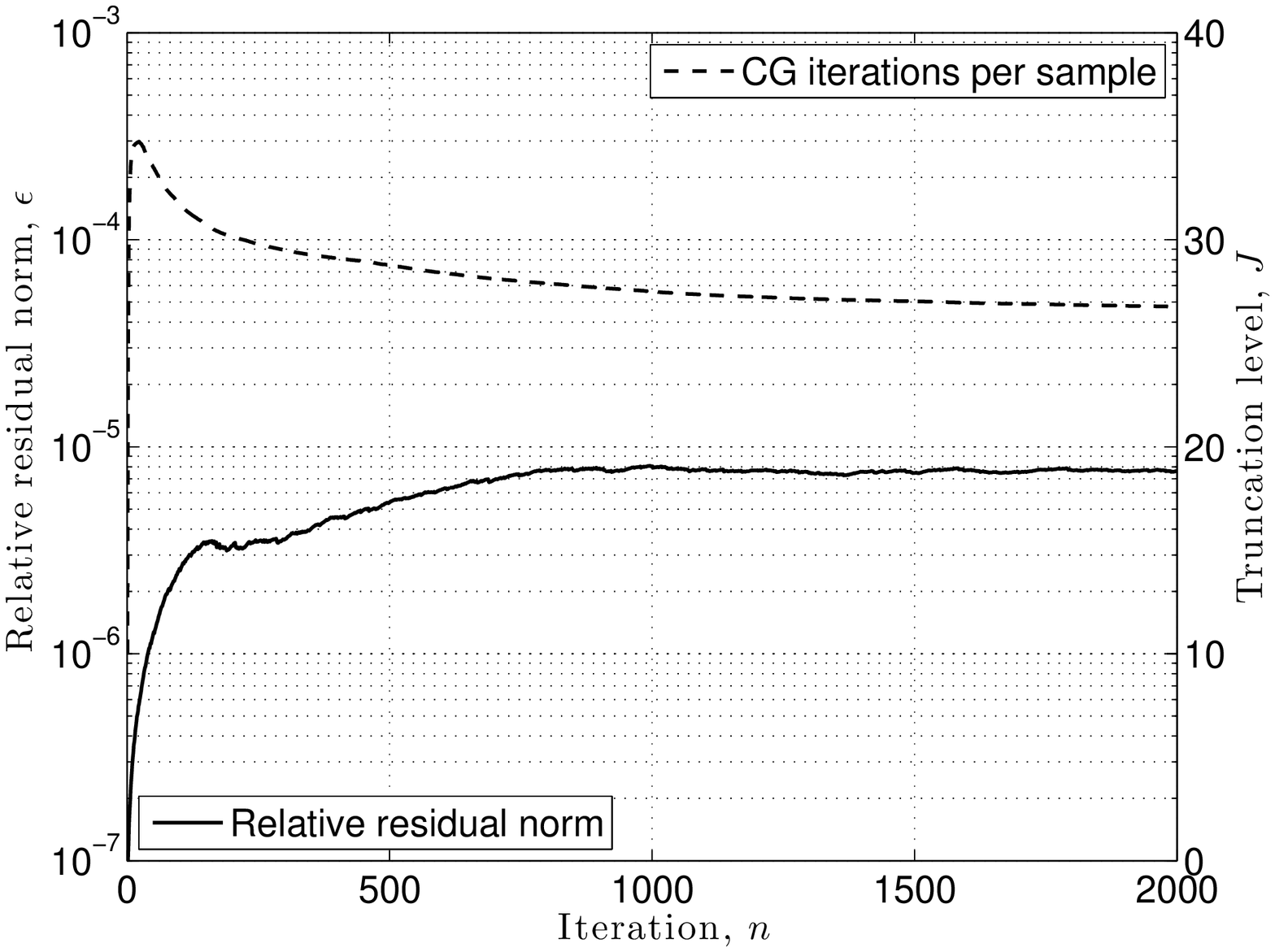}}
    \caption{(a) Influence of the CG truncation level on the overall computation cost and the statistical efficiency of the RJPO for sampling a Gaussian of dimension $N=128$. (b) Evolution of the relative residual norm and the acceptance rate for a Gaussian sampling problem of size $N=128$. The adaptive algorithm leads to a relative residual norm $\epsilon_{\textrm{opt}}=7.79 \cdot 10^{-6}$ leading to $\alpha_{\textrm{opt}}=0.977$ and $J_{\textrm{opt}}=26$.
    \label{fig:adrjpo}}
\end{figure}
\newpage

%
%

\section{Application to unsupervised super-resolution} \label{sec:appli}

In the linear inverse problem of  unsupervised image super-resolution, several images are observed with a low spatial resolution. In addition, the measurement process presents a point spread function that introduces a blur on the images. The purpose is then to reconstruct the original image with a higher resolution using an unsupervised method. Such an approach allows to also estimate the model hyper-parameters and the PSF~\cite{orieux2012sampling,park2003super,rochefort2006improved}. In order to discuss the relevance of the previously presented Gaussian sampling algorithms we apply a Bayesian approach and MCMC methods for solving this inverse problem. 

\subsection{Problem statement}
The observation model is given by $\yb~=~\Hb\xb+\nb$, where $\Hb = \Pb\Fb$, with $\yb\in\eR^M$ the vector containing the pixels of the observed images in a lexicographic order, $\xb\in\eR^N$ the sought high resolution image, $\Fb$ the $N\times N$ circulant convolution matrix associated with the blur, $\Pb$ the $M\times N$ decimation matrix and $\nb$ the additive noise.  

\paragraph{Statistical modeling.} The noise is assumed to follow a zero-mean Gaussian distribution with an unknown precision matrix $\Qb_y~=~\gamma_y\Ib$. We also assume a zero-mean Gaussian distribution for the prior of the sought variable $\xb$, with a precision matrix $\Qb_x=\gamma_x\Db\T\Db$. $\Db$ is the circulant convolution matrix associated to a Laplacian filter. Non-informative Jeffrey's priors~\cite{jeffreys1946invariant} are assigned to the two hyper-parameters $\gamma_y$ and $\gamma_x$.

\paragraph{Bayesian inference.} According to Bayes' theorem, the posterior distribution is given by
\begin{equation*}
P(\xb,\gamma_x,\gamma_y|\yb)\propto\gamma_x^{(N-1)/2-1}\gamma_y^{M/2-1}\times
e^{-\frac{1}{2}\gamma_y(\yb-\Hb\xb)\T(\yb-\Hb\xb)-\frac{1}{2}\gamma_x\xb\T\Db\T\Db\xb}
\end{equation*}
To explore this posterior distribution, a Gibbs sampler iteratively draws samples from the following conditional distributions:
\begin{enumerate}
\item $\gamma_y^{(n)}$ from $P\left(\gamma_y|\xb^{(n-1)},\yb\right)$ given as
$$
\Gc\left(1+\frac{M}{2}, 2 ||\yb-\Hb\xb^{(n-1)}||^{-2}\right),
$$
\item $\gamma_x^{(n)}$ from $ P\left(\gamma_x|\xb^{(n-1)}\right)$ given as
$$
\Gc\left(1+\frac{N-1}{2}, 2 ||\Db\xb^{(n-1)}||^{-2}\right)
$$
\item $\xb^{(n)}$ from $  P\left(\xb|\gamma_x^{(n)}, \gamma_y^{(n)}, \yb\right)$ which is
$$
\Nc\left(\mub^{(n)}, \left[\Qb^{(n)}\right]\M\right)
$$
with
\begin{align*}
\Qb^{(n)}&=\gamma_y^{(n)}\Hb\T\Hb+\gamma_x^{(n)}\Db\T\Db\\
\Qb^{(n)}\mub^{(n)}&=\gamma_y^{(n)}\Hb\T\yb
\end{align*}

\end{enumerate}
The third step of the sampler requires an efficient sampling of a multivariate Gaussian distribution whose parameters change along the sampling iterations. In the sequel, direct sampling with Cholesky factorization~\cite{rue2001fast} is firstly employed as a reference method. It yields the same results as the E-PO algorithm. For the inexact resolution case, the T-PO algorithm using a CG controlled by the relative residual norm, and the adaptive RJPO directly tuned with the acceptance probability are performed. For these two methods, the product matrix-vector used in the CG algorithm is done by exploiting the structure of the precision matrix $\Qb$ and thus only implies circulant convolutions, performed by FFT, and decimations. 

\subsection{Estimation results using MCMC}
We consider the observation of five images of dimension $128\times 128$ pixels ($M=81920$) and we reconstruct the original one of dimension $256\times 256$ ($N=65536$). The convolution part $\Fb$ has a Laplace shape with of full width at half maximum (FWHM) of 4 pixels. A white Gaussian noise is added to get a signal-to-noise ratio (SNR) equal to 20dB. The original image and one of the observations are shown in \figref{fig:sr:all}. 
\begin{figure}[h!]
\centering
{\small
\begin{tabular}{ccc}
 \includegraphics[width=0.29\columnwidth]{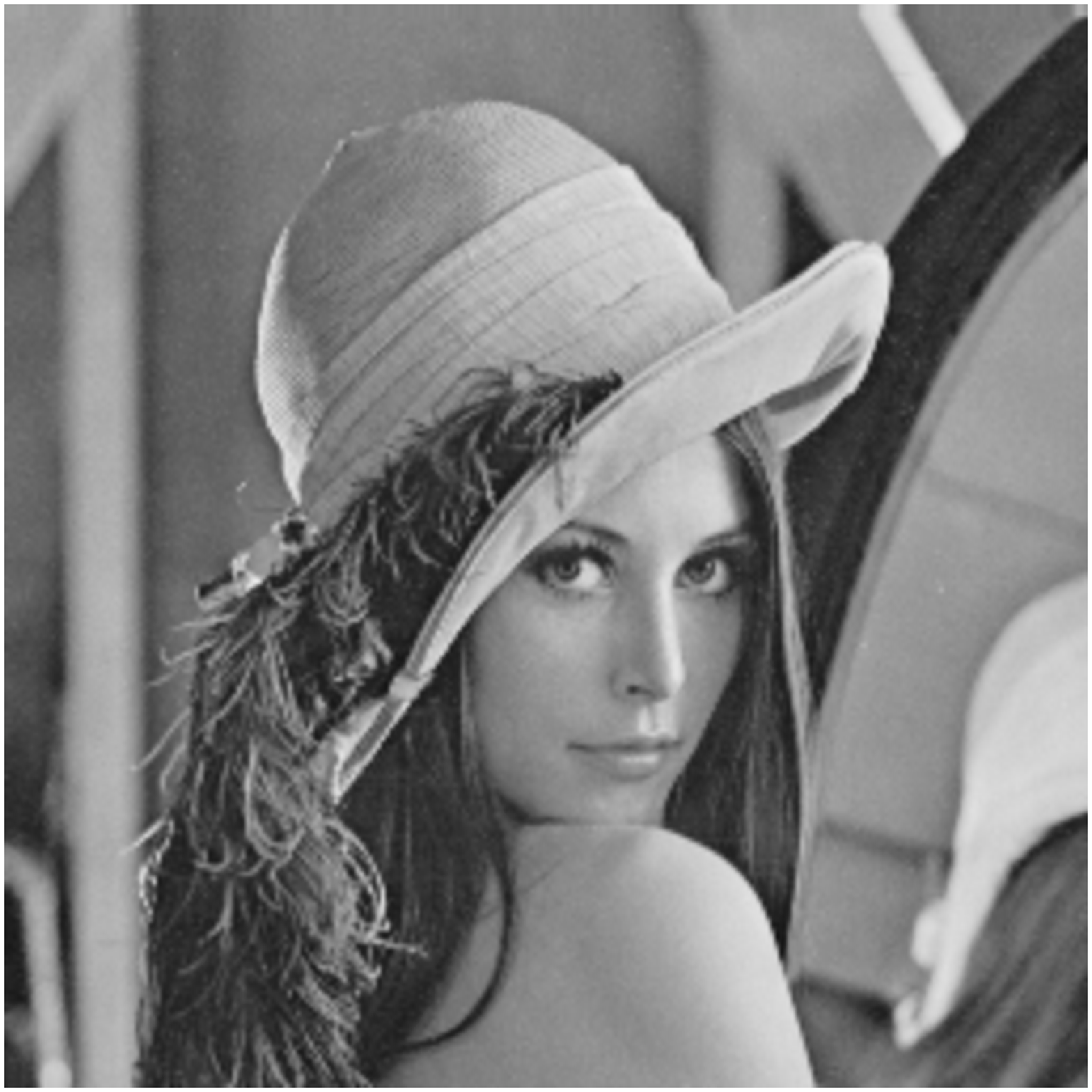} &
 \includegraphics[width=0.29\columnwidth]{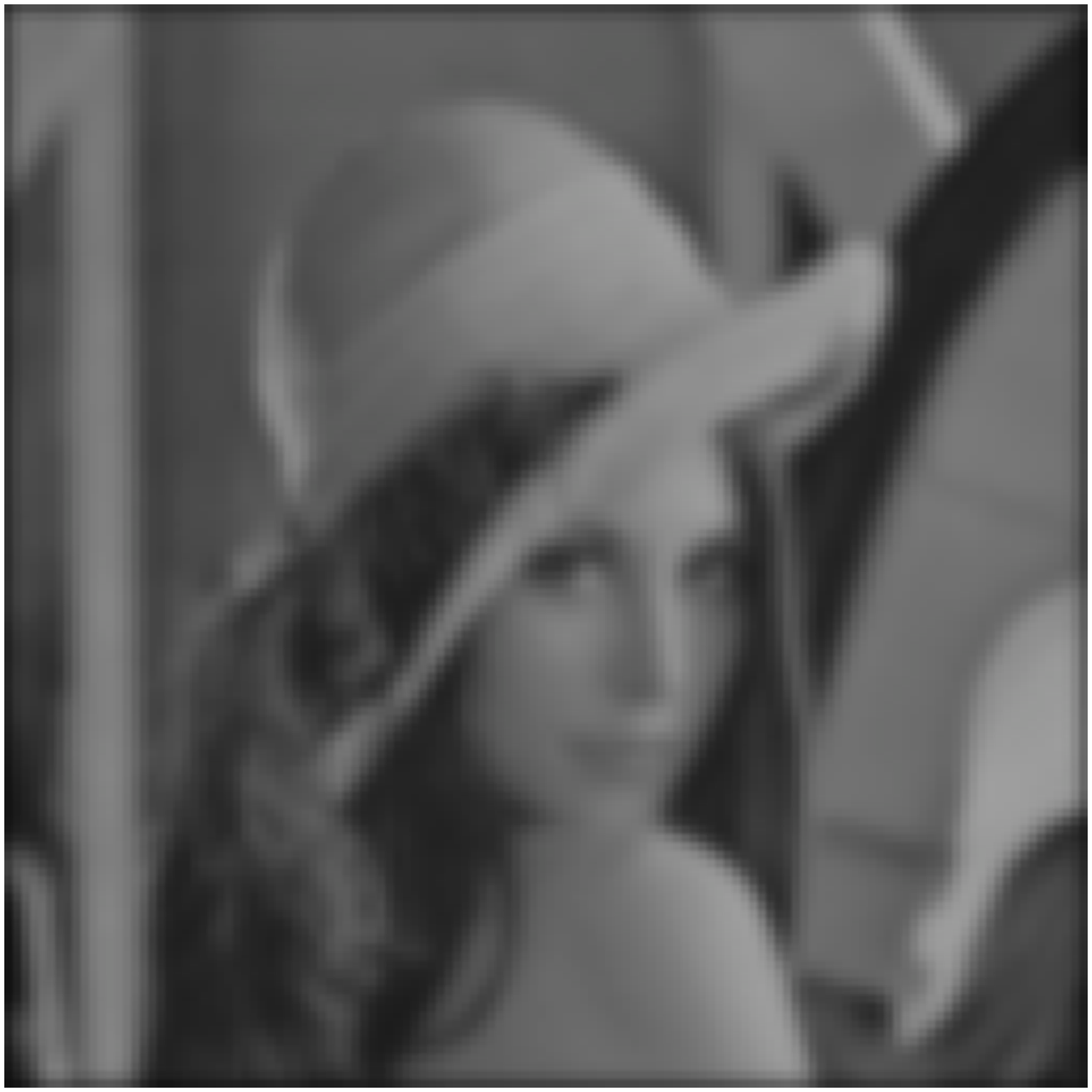}&
 \includegraphics[width=0.29\columnwidth]{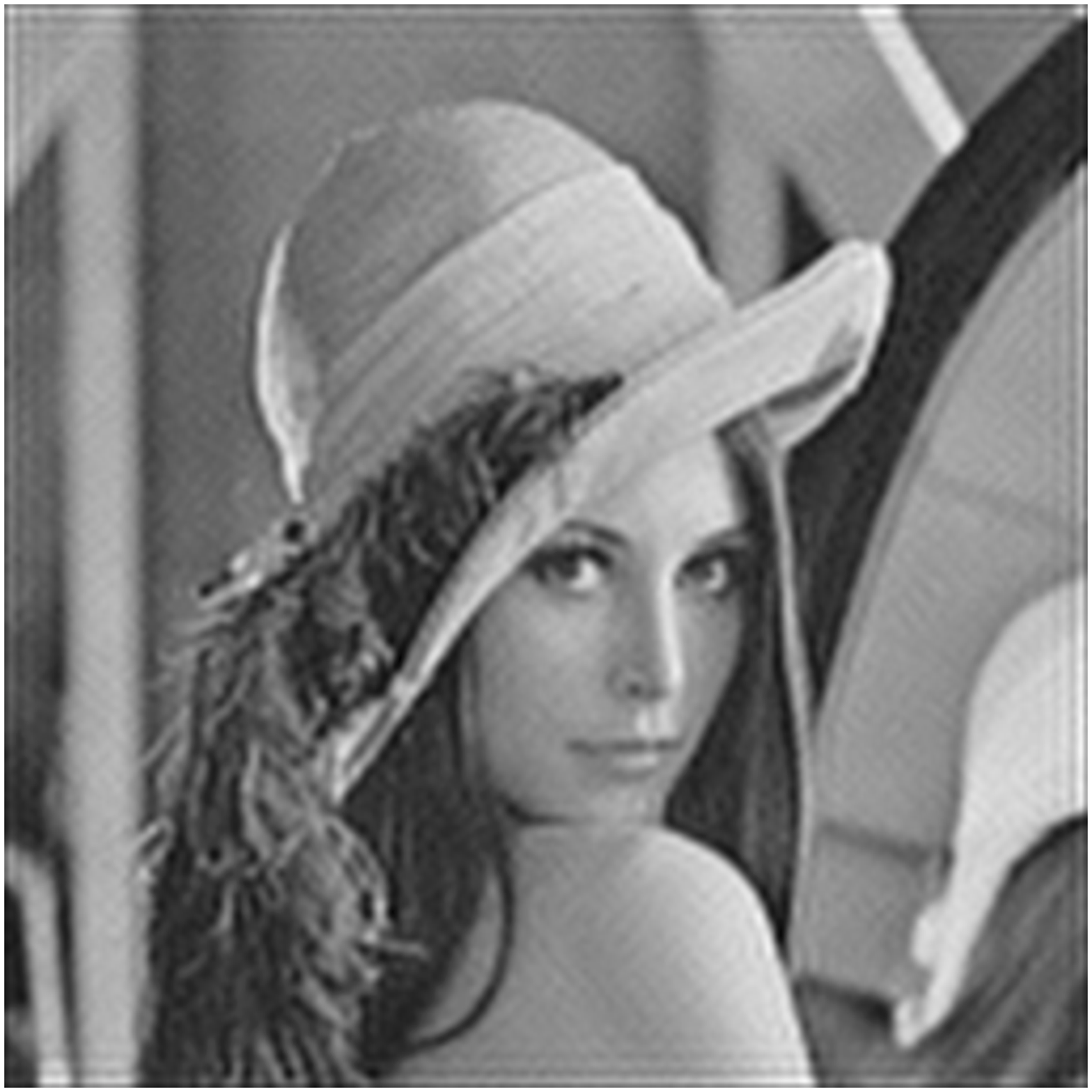}\\
Original & Observation & Reconstructed
\end{tabular}}
    \caption{%
        Unsupervised super-resolution - image reconstruction with adaptive RJPO algorithm and $\alpha_t=0.99$.}%
   \label{fig:sr:all}
\end{figure}

The Gibbs sampler is run for 1000 iterations and a burn-in period of 100 iterations is considered after a visual inspection of the chains. The performances are evaluated in terms of the mean and standard deviation of both hyper-parameters $\gamma_y$, $\gamma_x$ and one randomly chosen pixel $x_i$ of the reconstructed image. Table~\ref{tab:sr} presents the mean and standard deviation of the variable of interest. As we can see, the T-PO algorithm is totally inappropriate even with a precision of $10^{-8}$. Conversely, the estimation from the samples given by the adaptive RJPO and Cholesky method are very similar, which demonstrates the correct behavior of the proposed algorithm. 

\begin{table}[h!]
\centering
\begin{tabular}{|cc||c|c|c|}
\hline
&		& $\gamma_y$ & $\gamma_x\times10^{-4}$ & $x_i$ \\
\hline
\multicolumn{2}{|c|}{Cholesky} & 102.1~(0.56) & 6.1~(0.07) & 104.6~(9.06) \\
\hline
T-PO & $\epsilon=10^{-4}$ & 0.3~(0.06) & $45~(0.87)$ & 102.2~(3.30) \\
\hline
T-PO & $\epsilon=10^{-6}$ & 6.8~(0.04) & $32~(0.22)$ & 104.8~(2.34) \\
\hline
T-PO & $\epsilon=10^{-8}$ & 71.7~(0.68) & $21~(0.29)$ & 102.7~(2.51)\\
\hline
A-RJPO, & $\alpha_t=0.99$ & 101.2~(0.55)& $6.1~(0.07)$ & 101.9~(8.89) \\
\hline
\end{tabular}
\caption{Comparison between the Cholesky approach, the T-PO controlled by the relative residual norm and the A-RJPO tuned by the acceptance rate, in terms of empirical mean and standard deviation of hyper-parameters and one randomly chosen pixel.}\label{tab:sr}
\end{table}

\figref{fig:sr:alphas} shows the evolution of the average acceptance probability with respect to the number of CG iterations. We can notice that at least 400 iterations are required to have a nonzero acceptance rate. Moreover, more than 800 iterations seems unnecessary. For this specific problem, the E-PO algorithm needs theoretically $N=65536$ iterations to have a new sample while the adaptive RJPO only requires around 700. Concerning  the computation time, on a Intel Core i7-3770 with 8GB of RAM and a 64bit system, it took about 20.3s on average and about 6GB of RAM for the Cholesky sampler to generate one sample and only 15.1s and less than 200MB for the RJPO. This last result is due to the use of a conjugate gradient on which each matrix-vector product is performed without explicitly writing the matrix \Qb. Finally, note that if we consider images of higher resolution, for instance $N=1024\times1024$, the Cholesky factorization would require around 1TB of RAM and the Adaptive RJPO only about 3GB (when using double precision floating-point format).
\begin{figure}[h!]
\centering
\includegraphics[width=8cm]{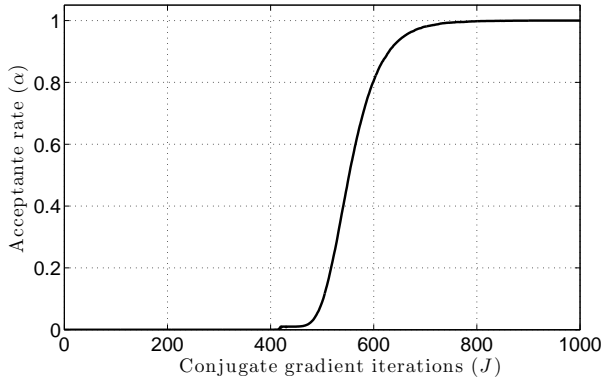}
\caption{Evolution of the acceptance rate with respect to average conjugate gradient iterations for sampling a Gaussian of dimension $N=65536$.}
\label{fig:sr:alphas}
\end{figure}
%
%
\section{Conclusion}
The sampling of high dimensional Gaussian distributions appears in the resolution of many linear inverse problems using MCMC methods. Alternative solutions to the Cholesky factorization are needed to reduce the computation time and to limit the memory usage. Based on the theory of reversible jump MCMC, we derived a sampling method allowing to introduce an approximate solution of a linear system during the sample generation step. The approximate resolution of a linear system was already adopted in methods like IFP and PO to reduce the numerical complexity, but without any guarantee of convergence to the target distribution. The proposed algorithm RJPO is based on an accept-reject step that is absent from the existing PO algorithms. Indeed, the difference between RJPO and existing PO algorithms is much comparable to the difference between the Metropolis-adjusted Langevin algorithm (MALA)~\cite{Roberts96}~and a plainly discretized Langevin diffusion.

Our results pointed out that the required resolution accuracy in these methods must be carefully tuned to prevent a significant error. It was also shown that the proposed RJ-MCMC framework allows to ensure the convergence through the accept-reject step whatever the truncation level. In addition, thanks to the simplicity of the acceptance probability, the resolution accuracy can be adjusted automatically using an adaptive scheme allowing to achieve a pre-defined acceptation rate. We have also proposed a significant improvement of the same adaptive tuning approach, where the target is directly formulated in terms of minimal computing cost per effective sample.

Finally, the linear system resolution using the conjugate gradient algorithm offers the possibility to implement the matrix-vector products with a limited memory usage by exploiting the structure of the forward model operators. The adaptive RJPO has thus proven to be less consuming in both computational cost and memory usage than any approach based on Cholesky factorization.

This work opens some perspectives in several directions. Firstly, preconditioned conjugate gradient or alternative methods can be envisaged for the linear system resolution with the aim to reduce the computation time per iteration. Such an approach will highly depend on the linear operator and the ability to compute a preconditioning matrix. A second direction concerns the connection between the RJ-MCMC framework and other sampling methods such as those based on Krylov subspace~\cite{aune2013iterative,parker2012sampling}, particularly with appropriate choices of the parameters \Ab, \Bb, \bb and $\fb(\cdot)$ defined in section~\ref{sec:revjump}. Another perspective of this work is to analyze more complex situations involving non-gaussian distributions with the aim to be able to formulate the perturbation step and to perform an approximate optimization allowing to reduce the computation cost. Finally, the proposed adaptive tuning scheme allowing to optimize the computation cost per effective sample could be generalized to other Metropolis adjusted sampling strategies.

\appendix

\section{Expression of the acceptance probability}
\label{annex_deltaS}
According to the RJ-MCMC theory, the acceptance probability is given by
\[
\alpha (\xbold, \xb|\zb) = \min\left(1,\frac{P_\Xb(\xb)P_\Zb(\sb|\xb)}{P_\Xb(\xbold)P_\Zb(\zb|\xbold)}|J_{\phib}(\xbold, \zb)|\right),
\]
with $\sb=\zb$ and $\xb = -\xbold + \fb(\zb)$. The Jacobian determinant of the deterministic move is $|J_{\phib}(\xb,\zb)|=1$.
Since
\[
P_\Xb(\xb)
\propto e^{-\frac{1}{2}(\xb-\mub)\T\Qb(\xb-\mub)},
\]
and
\[
P_\Zb(\zb|\xbold)
\propto
e^{-\frac{1}{2}(\zb-\Ab\xbold-\bb)\T\Bb\M(\zb-\Ab\xbold-\bb)},
\]
the acceptance probability can be written as $$\alpha(\xbold, \xb|\zb) =  \min\left(1,e^{-\frac{1}{2}\Delta S}\right)$$ with $ \Delta S = \Delta S_1+\Delta S_2$ and 
\begin{align*}
\Delta S_1&=(\xb-\mub)\T\Qb(\xb-\mub)-(\xbold-\mub)\T\Qb(\xbold-\mub),\\
&=\xb\T\Qb\xb-2\xb\T\Qb\mub-\xbold\T\Qb\xbold+2\xbold\T\Qb\mub.\\
\Delta S_2&=(\zb-\Ab\xb-\bb)\T\Bb\M(\zb-\Ab\xb-\bb) -(\zb-\Ab\xbold-\bb)\T\Bb\M(\zb-\Ab\xbold-\bb),\\
&=\xb\T\Ab\T\Bb\M\Ab\xb-2\xb\T\Ab\T\Bb\M(\zb-\bb) -\xbold\T\Ab\T\Bb\M\Ab\xbold+2\xbold\T\Ab\T\Bb\M(\zb-\bb).
\end{align*}
Since $\xb=-\xbold+\fb(\zb)$, we get
\begin{align*}
\Delta S_1&=(\xb-\xbold)\T\Qb\left(\fb(\zb)-2\mub\right)\\
\Delta S_2&=(\xb-\xbold)\T\left(\Ab\T\Bb\M\Ab \, \fb(\zb)-2\Ab\T\Bb\M(\zb-\bb)\right)
\end{align*}
Finally
\begin{align*}
\Delta S&=\left(\xb-\xbold\right)\T
\FINAL{\\&}
\left[\left(\Qb+\Ab\T\Bb^{-1}\Ab\right)\fb(\zb)-2\left(\Qb\mub+\Ab\T\Bb^{-1}\left(\zb-\bb\right)\right)\right]
\FINAL{\\&}
=2 \,  \left(\xbold-\xb\right)\T\rb(\zb).
\end{align*}
Finally, when the system is solved exactly, $\Delta S = 0$ and thus $\alpha(\xbold, \xb|\zb)=1$. 

\section{Correlation between two successive samples}
\label{annex_corr}
Since
\begin{align*}
\xbnew=-\xbold+2\left(\Qb+\Ab\T\Bb^{-1}\Ab\right)^{-1}
\left(\Qb\mub +\Ab\T\Bb^{-1}(\zb-\bb)\right)
\end{align*}
and  $\zb$ is sampled from $\Nc\left(\Ab\xbold+\bb, \Bb\right)$, we have
\begin{align*}
\xbnew=-\xbold+2&\left(\Qb+\Ab\T\Bb^{-1}\Ab\right)^{-1}
\FINAL{\\&}
\left(\Qb\mub +\Ab\T\Bb^{-1}\Ab\xbold+\Ab\T\Bb\M\omegab_{\Bb}\right)
\end{align*}
with $\omegab_{\Bb}$ totally independent of $\xbold$. One can firstly check that $\mathbb{E}\left[\xbold\right]=\mathbb{E}\left[\xbnew\right]=\mub$. Consequently, the correlation between two successive samples is given by
\begin{equation*}
\mathbb{E}\left[(\xbnew-\mub)(\xbold-\mub)\T\right]=\FINAL{\\}\left(2\left(\Qb+\Ab\T\Bb^{-1}\Ab\right)^{-1}\Ab\T\Bb^{-1}\Ab-\Ib\right)\Qb\M
\end{equation*}
which is zero if and only if $\Ab\T\Bb\M\Ab=\Qb$.

\bibliographystyle{plain}

\bibliography{revueabr,Gilavert2014rjpo} 

 \end{document}